\documentclass[9pt,shortpaper,twoside,web]{IEEEtran}
\usepackage{generic}
\usepackage{silence}
\usepackage[noadjust]{cite}
\usepackage{amsmath,amssymb,amsfonts}
\usepackage{amsthm}
\usepackage{algorithmic}
\usepackage{algorithm}
\usepackage{graphicx}
\usepackage{graphbox}
\usepackage{textcomp}
\usepackage{multirow}
\usepackage{balance}
\usepackage{xcolor}
\usepackage[export]{adjustbox}
\usepackage{color}
\usepackage{subcaption}
\usepackage{accents}
\usepackage{textcomp}
\usepackage{comment}
\usepackage{bm}
\usepackage{tikz}
\let\labelindent\relax
\usepackage[shortlabels]{enumitem}
\def\BibTeX{{\rm B\kern-.05em{\sc i\kern-.025em b}\kern-.08em
    T\kern-.1667em\lower.7ex\hbox{E}\kern-.125emX}}

\newcommand{\blue}[1]{\color{black}{#1}\color{black}\phantom{}}
\definecolor{LightCyan}{rgb}{0.85,0.85,1.0}
\newcommand\myfontsize{\fontsize{9pt}{9pt}\selectfont}

\newtheorem{theorem}{Theorem}

\newtheorem{lemma}{Lemma}
\newtheorem{remark}{Remark}
\newtheorem{definition}{Definition}
\newtheorem{proposition}{Proposition}

\newtheorem{problem}{Problem}

\makeatletter
\let\NAT@parse\undefined
\makeatother
\usepackage[hidelinks]{hyperref}

\begin{document}

\title{ECO-DKF: Event-Triggered and Certifiable Optimal \\ Distributed Kalman Filter under Unknown Correlations}
\author{Eduardo Sebasti\'{a}n*,  Eduardo Montijano and Carlos Sagü\'{e}s
\thanks{This work has been supported by Spanish projects PID2021-125514NB-I00, PID2021-124137OBI00 and TED2021-130224B-I00 funded by MCIN/AEI/10.13039/501100011033, by ERDF A way of making Europe and by the European Union NextGenerationEU/PRTR, DGA T45-23R, and Spanish grant FPU19-05700.}
\thanks{E. Sebasti\'{a}n, E. Montijano and C. Sagü\'{e}s are associated with the Instituto de Investigaci\'on en Ingenier\'ia de Arag\'on, Universidad de Zaragoza, Spain 
\texttt{\small \{esebastian, emonti, csagues\}@unizar.es}. *Corresponding author.}
}

\newcommand\copyrighttext{%
  \footnotesize \textcopyright This paper has been accepted for publication at IEEE Transactions on Automatic Control. Please, when citing the paper, refer to the official manuscript published by the IEEE Transactions on Automatic Control.}
\newcommand\copyrightnotice{%
\begin{tikzpicture}[remember picture,overlay]
\node[anchor=south,yshift=10pt] at (current page.south) {\fbox{\parbox{\dimexpr\textwidth-\fboxsep-\fboxrule\relax}{\copyrighttext}}};
\end{tikzpicture}%
}

\maketitle
\copyrightnotice
\thispagestyle{empty}
\pagestyle{empty}

\begin{abstract}
This paper presents ECO-DKF, the first \textit{E}vent-Triggered and \textit{C}ertifiable \textit{O}ptimal \textit{D}istributed \textit{K}alman \textit{F}ilter. Our algorithm addresses two major issues  
inherent to
Distributed Kalman Filters: (i) fully distributed and scalable optimal estimation and (ii) reduction of the communication bandwidth usage. The first requires to solve an NP-hard optimisation problem, forcing relaxations that lose optimality guarantees over the original problem. Using only information from one-hop neighbours, we propose a tight Semi-Definite Programming relaxation that allows to certify locally and online if the relaxed solution is the optimum of the original NP-hard problem. In that case, ECO-DKF is optimal in the square error sense under scalability and event-triggered one-hop communications restrictions. Additionally, ECO-DKF is a globally asymptotically stable estimator. To address the second issue, we propose an event-triggered scheme from the relaxed optimisation output. The consequence is a broadcasting-based algorithm that saves communication bandwidth, avoids individual communication links and multiple information exchanges within instants, and preserves the optimality and stability properties of the filter.  
\end{abstract}

{\myfontsize }


\section{Introduction}\label{sec:intro}
The Kalman Filter (KF)~\cite{Kalman1960Kalman} is a cornerstone in control theory. Its elegance and optimality has motivated its extension towards distributed setups.  
However, two major issues prevent a direct deployment in real setups. First, the conservation of the KF optimality. Second, the efficient usage of the communication bandwidth.

Regarding optimality, the works by Olfati-Saber~\cite{Olfati2005DKF,Olfati2007DKF} opened an era of consensus-based Distributed Kalman Filters (DKF). However, it is proved that the computation of the estimator's gains is not scalable~\cite{Olfati2009DKF, Zhu2013Consensus}. Since then, several works have aroused searching for an Optimal DKF (O-DKF)~\cite{Kamal2013Information, Battistelli2014Consensus, Das2016Consensus} by either keeping track of the whole network error covariance matrix or neglecting the information encoded in the covariance error matrices. In a recent work~\cite{Battilotti2021MultipleConsensus} it is discussed that, to reach optimal centralised KF performance, the state-of-the-art DKFs typically resort to multiple consensus steps among instants. The absence of solutions which optimally and \textit{scalably} integrate the neighbouring estimates and error covariances motivated the development of Diffusion KFs (DfKF). The most popular approach is the Covariance Intersection (CI) method~\cite{Julier1997CI,Julier2001CI,Hu2011Diffusion}. The CI method is generally suboptimal~\cite{Reinhardt2015CI}, thus arising a necessity of certification. The solution to approach optimality goes through convex optimisation problems~\cite{Cattivelli2010Diffusion}, sequential algorithms~\cite{Deng2012sequential} or diffusion steps that require global topology knowledge~\cite{Zhang2015Diffusion}.
Inspired by the DfKF solutions, ECO-DKF gathers the benefits of all the aforementioned solutions and adds the \textit{certifiability} property, which comes from a reformulation of the CI method based on the outer Löwner-John ellipsoid intersection method~\cite{Boyd1994LMI} and a Semi-Definite Programming (SDP) relaxation. By proving the equivalence between ECO-DKF and the optimal consensus over estimates we can demonstrate global asymptotic stability, independent on any hand-tuned gain. 

Regarding communications, the amount of information sent through the network is related to its energy consumption and the available bandwidth~\cite{Wu2012Event}, so an Event-Triggered (ET) estimator is desirable. The most popular ET rule is the Send-on-Delta (SoD)~\cite{Suh2007Modified}, which computes the difference between estimates \cite{Li2015Event,Liu2015Fading,Liu2021Event,Zhang2017Distributed}, measurements \cite{Suh2007Modified,Liu2021Event,Zhang2017Distributed} and/or innovations \cite{You2012Kalman,Liu2015Event,Battistelli2018Distributed} to trigger the information when the difference exceeds some tuned threshold. SoD is popular because it is easy to compute and agnostic to the employed estimator. In contrast, ECO-DKF exploits the estimator's structure to derive an ET rule which is even easier to compute and does not require tuning. A work conducted by Trimpe and Campi~\cite{Trimpe2015Choice} compared the most popular ET rules, concluding that leveraging the error covariance in the decision improves performance, e.g., the Kullback-Leibler (KL) divergence. Regarding KFs, ET Centralised Kalman Filters (CKF)~\cite{Battistelli2012Data,Wu2012Event,You2012Kalman,Han2013Stochastic,Trimpe2014Event} are subject to a critical arrival rate beyond which the error covariance becomes unbounded~\cite{Sinopoli2004Kalman}, underlining the importance of a carefully designed ET protocol. On the other hand, both the SoD~\cite{Battistelli2018Distributed,Zhang2017Distributed,Liu2021Event} and KL~\cite{Battistelli2018Distributed,Battistelli2019Event} protocols are very popular in non-optimal DKFs. 
In contrast, ECO-DKF is event-triggered-by-construction. This allows us to define an ET rule which is broadcasting-based, inexpensive to compute and avoids individual communication links and multiple information exchanges within instants. 

Exploiting these ideas, our main contribution is ECO-DKF, the \textit{first DKF with certifiability guarantees and event-triggered by construction}.  
If the certification is positive, then ECO-DKF is \textit{optimal in the square error sense} under unknown correlations and potential one-hop event-triggered communications. The proposed algorithm is also \textit{globally asymptotically stable} under mild periodic joint connectivity assumptions. Besides, nodes can \textit{check optimality locally and in real time}. In particular, ECO-DKF exploits the optimisation to derive a broadcasting-based rule, yielding to an algorithm which is \textit{event-triggered by construction}. ECO-DKF is \textit{fully distributed}, works with \textit{heterogeneous} sensor models and does not require \textit{parameter tuning}. 

This paper is an extension of~\cite{Sebastian2021CDC}. Compared to it, we first improve ECO-DKF from being time-triggered to event-triggered. This leads to new stability and optimality results that hold for mild connectivity settings. Besides, we relax the observability conditions to that of network-observability, a weaker assumption. Consistency of ECO-DKF is proved for the first time, while Lemmas that were presented in~\cite{Sebastian2021CDC} are now proved, including additional corollaries. The ET scheme is also a novel contribution, along with the corresponding convergence and optimality guarantees associated to it. Finally, the empirical evaluation of ECO-DKF is improved by adding more time-triggered state-of-the-art DKFs and including new experiments that validate the event-triggered properties of ECO-DKF.


\section{Problem formulation}\label{sec:preliminaries}
The target system is described by linear dynamics
\begin{equation}\label{eq:system}
    \mathbf{x}(k+1) = \mathbf{A}\mathbf{x}(k) + \mathbf{w}(k),
\end{equation}
where $k \in \mathbb{N}_{\ge 0}$ denotes the discrete time, $\mathbf{x} \in \mathbb{R}^{n}$ is the state of the system to estimate, $\mathbf{A} \in \mathbb{R}^{n \times n}$ is a matrix comprising the dynamics of the target system, and $\mathbf{w}(k) \sim \mathcal{N}(\mathbf{0},\mathbf{Q})$ is a white process. The system is tracked by a network of sensors described by a directed communication graph $G(k) = (V,E(k))$. The number of nodes is $N=|V|$, where $|\cdot|$ denotes cardinality. $\mathcal{N}_i(k) = \{j|(i,j)\in E(k) \}$ is the set of neighbours of node $i$ at instant $k$ and $\mathcal{J}_i(k) = \mathcal{N}_i(k) \cup i$. We use $\lfloor \cdot \rfloor_{ij}$ to denote the $ij$-th element or block component of a matrix. Sensors are described by a linear model
\begin{equation}\label{eq:distributedSensorModel}
    \mathbf{z}_i(k) = \mathbf{H}_i\mathbf{x}(k) + \mathbf{v}_i(k).
\end{equation}
In this expression, $\mathbf{z}_i \in \mathbb{R}^{m_i}$ is the measurement of node $i$, $\mathbf{H}_i \in \mathbb{R}^{m_i \times n}$ is the (unbiased) sensor model of node $i$, $m_i$ is the dimension of $\mathbf{z}_i$ and $\mathbf{v}_i(k) \sim \mathcal{N}(\mathbf{0},\mathbf{R}_i)$ is a white process. The measurements are independent between nodes. Notice that we are considering heterogeneous sensors.
Unless unclear, from now on we will omit the dependencies with $k$. Regarding observability, we assume \textit{network-observability}, defined as follows.
\begin{definition}\label{definition:observability}
System~\eqref{eq:system} is network-observable if the pair $(\mathbf{A},[\mathbf{H}_1^T,\mathbf{H}_2^T,\hdots,\mathbf{H}_N^T]^T)$ is observable.
\end{definition}

The objective of the network is to cooperatively estimate the state of system~\eqref{eq:system}. Each node has an estimate of the state of system~\eqref{eq:system}, $\hat{\mathbf{x}}_i$, with associated error covariance matrix $\hat{\mathbf{P}}_i := \text{E}[(\hat{\mathbf{x}}_i-\mathbf{x})(\hat{\mathbf{x}}_i-\mathbf{x})^T]$. Besides, similarly to other typical filter algorithms, the nodes use an auxiliary variable for the prediction stage, $\bar{\mathbf{x}}_i,$ and $\bar{\mathbf{P}}_i := \text{E}[(\bar{\mathbf{x}}_i-\mathbf{x})(\bar{\mathbf{x}}_i-\mathbf{x})^T] $. In the classical CKF, the estimate $\hat{\mathbf{x}}_i(k-1), \hat{\mathbf{P}}_i(k-1)$ is first used to obtain the prediction $\bar{\mathbf{x}}_i(k), \bar{\mathbf{P}}_i(k)$ by propagation through the target dynamics. Then, the measurement is used to correct the prediction, leading to the updated estimate $\hat{\mathbf{x}}_i(k), \hat{\mathbf{P}}_i(k)$. The CKF is not scalable in networked applications because a central computation unit is needed to aggregate the measurements from all the nodes. An alternative is to develop a DKF, where each node uses only local and neighbouring information. 

\blue{Given, $\bar{\mathbf{P}}_i(k)$, we define $\bar{\mathbf{S}}_i(k) = \bar{\mathbf{P}}_i^{-1}(k)$ as the predicted error covariance matrices in information form. In this paper we are interested in viewing these matrices as ellipsoids, defined as follows. 
\begin{definition}\label{definition:ellipsoid}
Given $\bar{\mathbf{S}}_i$ and assuming unbiased sensors, the ellipsoid $\varepsilon_i^i$ is $ \varepsilon_i^i := \{\,\mathbf{x} \,| \,\mathbf{x}^T \bar{\mathbf{S}}_i^{-1} \mathbf{x} \leq 1 \, \}$.
\end{definition}
}
In our proposal, this absence of global knowledge leads to an NP-hard optimisation problem, caused by the unknown correlations between the nodes' estimates,
forcing the use of convex relaxations to solve it. Despite enabling tractability, the solution of the relaxed problem is not guaranteed to be the optimum of the original problem. This motivates the need of certifiability on the optimisation, formally defined as follows:
\begin{definition}[\textbf{From Definition 19 in~\cite{Yang2020Teaser}}]\label{definition:certifiability}
Given an optimization problem $\mathbb{O}(\mathbb{D})$ that depends on input data $\mathbb{D}$, we say that an algorithm $\mathbb{A}$ is certifiable if, after solving $\mathbb{O}(\mathbb{D})$, $\mathbb{A}$ either provides a certificate for the optimality of its solution or declares failure otherwise.
\end{definition}

Besides, to avoid divergence of the estimation, the result of the optimisation must be consistent, defined as follows:
\begin{definition}\label{definition:consistency}
An approximation $\tilde{\mathbf{P}}$ of $\mathbf{P}$ and $\tilde{\mathbf{x}}$ of $\mathbf{x}$ is consistent if and only if $\tilde{\mathbf{P}} - \mathbf{P} \succeq \mathbf{0}$ and $\textnormal{E}[(\tilde{\mathbf{x}}-\mathbf{x})] = \mathbf{0}$,
where $\succeq$ means positive semi-definiteness.
\end{definition}

Many DKFs rely on time-triggered communications of all the nodes, requiring an intensive bandwidth usage. Instead, event-triggered communications save energy and bandwidth. In this regard, we define two concepts. The first is Periodic Joint Connectivity (PJC):
\begin{definition}\label{definition:PJC}
Consider a sequence of graphs $\{G(k), G(k+i),\hdots, G(k+t) \}$ with $t \in \mathbb{N}$. If $\mathbf{G}(k,t)=\cup_{i=0}^{t} G(k+i)$ is a strongly connected graph, then $\mathbf{G}(k,t)$ is a periodic joint connected graph with period $t$ at instant $k$.
\end{definition}
Second, we define a metric to measure the communication usage in the network. Before, we define the concept of ``triggering function'':

\begin{definition}\label{definition:triggering_function}
A triggering function $g_i(\cdot, k) \in \{0, 1\}$ is an indicator function that depends only on local information at node $i$ and such that, if $g_i(\cdot, k)=1$, then node $i$ broadcasts at instant $k$, and, if $g_i(\cdot, k)=0$, then node $i$ does not broadcast at instant $k$.
\end{definition}

\begin{definition}\label{definition:NoB}
The Number of Broadcasts ({\normalfont NoB}) at instant $k$ in network $G(k)$ is
\begin{equation}\label{eq:def_nob}
    \textnormal{NoB}(k) := \sum_{i=1}^{N} g_i(\cdot, k).
\end{equation}
\end{definition}
We now formulate the problem addressed in the paper.

\begin{problem}\label{problem:ECO-DKF}
Find an event-triggered stable algorithm that certificates locally and in real time if, at each $k$, each node $i$ minimises $\textnormal{E}[||\hat{\mathbf{x}}_i(k)-\mathbf{x}(k)||^2]$, under the following restrictions:
\begin{enumerate}
    \item \textbf{Locality}: At instant $k$ and $\forall i$, node $i$ only uses $\mathbf{A}$, $\mathbf{H}_i$, $\mathbf{Q}$, and $\mathbf{R}_i$ as parameters. From instant $k$ to $k+1$ and $\forall i$, node $i$ stores $\hat{\mathbf{x}}_i(k)$, $\hat{\mathbf{P}}_i(k)$, $\mathbf{A}$, $\mathbf{Q}$,  $\mathbf{H}_i$ and $\mathbf{R}_i$. 
    \item \textbf{One-hop communication}: At instant $k$ and $\forall i$, node $i$ communicates at most once with its neighbours $j \in \mathcal{N}_i(k)$. 
    \item \textbf{Consistency and connectivity}: The ET rule ensures consistency $\forall i, k$, and PJC of $G(k)$ $\forall k$ and some finite $t$.
    \item \textbf{Network-observability}.
\end{enumerate}
\end{problem}
\begin{remark}
From Problem~\ref{problem:ECO-DKF}, each node seeks to minimize, at each instant $k$, $\hbox{E}[|| \hat{\mathbf{{x}}}_i(k)-\mathbf{{x}}(k)||^2]$. Therefore, globally, the network minimises, at each instant $k$, $\sum_i \hbox{E}[|| \hat{\mathbf{{x}}}_i(k)-\mathbf{{x}}(k)||^2]$. Note that this definition holds because the minimisation refers to the current instant rather than a finite or infinite horizon of time.   
\end{remark}
We will use ${}^*$ to denote optimality in the sense of Problem~\ref{problem:ECO-DKF}. We assume that communication delays, quantisation effects and dropouts are negligible, a common assumption in many works (\cite{Shi2014Set,Liu2015Event,Battistelli2012Data,You2012Kalman} among others).


\section{ECO-DKF Algorithm}\label{sec:algorithm}
The proposed solution for Problem~\ref{problem:ECO-DKF} is the ECO-DKF algorithm. We first offer an overview of ECO-DKF and then, in Sections~\ref{sec:solution},~\ref{sec:stability} and~\ref{sec:event_triggered}, we formally study the main properties of the estimator. 

ECO-DKF algorithm is based on a novel optimization problem to aggregate predicted KF estimates. We propose to use the outer Löwner-John (LJ) method~\cite{John2014Ellipsoid},
which computes the smallest ellipsoid that surrounds the intersection of a set of ellipsoids,
\begin{subequations}\label{eq:outerLJellipse}
\begin{alignat}{2}
\bar{\mathbf{S}}_i^*, \bm{\lambda}^*_i =  \:\:\: &\underset{\bar{\mathbf{S}},\bm{\lambda}}{\arg\min}         \:\:\:\:\:\:\:\:\hbox{Tr}(\bar{\mathbf{S}}^{-1})  \label{eq:optProb3}
\\
   \:\:\:\:\:\: s.t.  \:\:\:\:\:&\mathbf{0}  \prec \bar{\mathbf{S}} \preceq \sum_{j=1}^{|\mathcal{J}_i|} \lambda_{j} \bar{\mathbf{S}}_j ,\label{eq:constraint31}
\\
&\sum_{j=1}^{|\mathcal{J}_i|} \lambda_{j}   \leq 1, \text{ }
\lambda_{j}  \geq 0 \:\:\forall j \in \mathcal{J}_i, \label{eq:constraint32}
\end{alignat}
\end{subequations}
where \blue{$\bar{\mathbf{S}}_j = \bar{\mathbf{P}}_j^{-1}$ characterise the ellipsoids}, Tr$(\cdot)$ is the trace of a matrix, $\lambda_{j}$ is the $j$-th element of vector $\bm{\lambda}$, and $\prec$ and $\preceq$ denote definiteness and semi-definiteness.
The selection of the trace as the optimisation cost function follows from minimising the square error $\textnormal{E}[||\hat{\mathbf{x}}_i-\mathbf{x}||^2]$ in Problem~\ref{problem:ECO-DKF}, proved later on Theorem~\ref{theorem:optimality}. The input of the outer LJ method is the predicted covariance, in information form, from node $i$ and its neighbours. 
Meanwhile, $\bm{\lambda}_i$ weights the importance of each prediction. The output of~\eqref{eq:outerLJellipse} is used to aggregate the predictions as follows
\begin{flalign}
    \bar{\mathbf{x}}_{i}^* = \bar{\mathbf{P}}_{i}^* \kern -0.2cm \sum_{j\in \mathcal{J}_i(k)}\kern -0.2cm \lambda_{ij}^*\bar{\mathbf{s}}_j
    \hbox{ and }
    \bar{\mathbf{P}}^*_{i} = (\bar{\mathbf{S}}_i^*)^{-1} \label{eq:KF_o1}.
\end{flalign}
The study of this aggregation is developed in Section~\ref{sec:solution}, including the certification guarantees on~\eqref{eq:outerLJellipse} which permits the formal results about optimality in Section~\ref{sec:stability}. To the best of our knowledge, this is the first time optimisation problem~\eqref{eq:outerLJellipse} is applied in a DKF. 

Interestingly, the output of~\eqref{eq:outerLJellipse} allows to develop a novel ET rule to decide if a node broadcasts a message to its neighbours:  
\begin{equation}\label{eq:ETrule1}
\kern -0.3cm\begin{aligned}
    &g_i(\bm{\lambda}_i^*(k-1),k) \equiv \\&g_i(\bm{\lambda}_i^*) = \left\{
\begin{aligned}
    1, & \hbox{ if } (\hbox{max}(\bm{\lambda}_i^*) = \lambda_{ii}^*) \vee (\mathcal{J}_i = \{i\} \hbox{ } \kern -0.1cm \wedge \kern -0.1cm \hbox{ } \text{Ber}(p) = 1 )
    \\ 
    0, & \hbox{ otherwise }
\end{aligned}
\right.
\end{aligned}
\end{equation}
where $\lambda_{ii}^* \in [0,1]$ is the weight of $\bm{\lambda}_i^*$ associated to node $i$, $\vee$ is the logical ``OR'', $\wedge$ is the logical ``AND'', and $\hbox{Ber}(p)$ represents the Bernoulli distribution of probability $p$. 
The messages are defined by $\hbox{msg} = \{ \mathbf{U}_i, \mathbf{u}_i, \bar{\mathbf{S}}_i, \bar{\mathbf{s}}_i  \}$, where
\begin{equation*}
    \begin{aligned}
    \mathbf{u}_i = \mathbf{H}_i^T\mathbf{R}_i^{-1}\mathbf{z}_i, \kern 0.1cm
    \mathbf{U}_i = \mathbf{H}_i^T\mathbf{R}_i^{-1}\mathbf{H}_i, \kern 0.1cm
    \bar{\mathbf{s}}_i = \bar{\mathbf{P}}_i^{-1}\bar{\mathbf{x}}_i, \kern 0.1cm \bar{\mathbf{S}}_i = \bar{\mathbf{P}}_i^{-1}.
    \end{aligned}
\end{equation*}
The quantities are in information form~\cite{anderson2012optimal} to comply with restriction 1 of Problem~\ref{problem:ECO-DKF}. 
We use the term ``broadcast'' because, when a node sends its information, there is no distinction to which neighbours the message arrives. A selective communication requires establishment of particular communication channels and a priori knowledge on the set of neighbours. Besides, we do not allow multiple communications within an instant.

Now, we are ready to present ECO-DKF in Algorithm~\ref{al:ECO-DKF}. Given $\hat{\mathbf{x}}_i^*(k-1)$, $\hat{\mathbf{P}}_i^*(k-1)$ and $\bm{\lambda}_i^*(k-1)$, each node first computes the prediction of the state through the target system dynamics, \begin{flalign}
    \bar{\mathbf{P}}_i(k) &= \mathbf{A}\hat{\mathbf{P}}_i^*(k-1)\mathbf{A}^T + \mathbf{Q}\label{eq:KF5},
    \\
    \bar{\mathbf{x}}_i(k) &= \mathbf{A}\hat{\mathbf{x}}_i^*(k-1)\label{eq:KF6}.
\end{flalign}
Then, each node obtains its measurement $\mathbf{z}_i(k)$ and decides whether to broadcast a message or not according to the ET rule~\eqref{eq:ETrule1}. Node $i$ receives messages from broadcasting neighbours. The information from the messages is aggregated to compute the values that will be used in the correction step. In the case of the measurements, this is direct because they are independent
\begin{flalign}
    \mathbf{y}_i = \sum_{j\in \mathcal{J}_i} \mathbf{u}_j \hbox{ and }
    \mathbf{Y}_i =\sum_{j\in \mathcal{J}_i} \mathbf{U}_j\label{eq:KF1}.
\end{flalign}
On the other hand, the aggregation of predictions is harder because we must consider the unknown correlations between them.
To solve this, we use~\eqref{eq:outerLJellipse} and~\eqref{eq:KF_o1}.

Finally, each node calculates the estimated error covariance matrix
\begin{equation}
    \hat{\mathbf{P}}_{i}^*= \left(\bar{\mathbf{S}}_i^* + \mathbf{Y}_i\right)^{-1}\label{eq:KF3}.
\end{equation}
This is analogous to the correction step of the information form~\cite{anderson2012optimal} DKF presented by Olfati-Saber (Eq. (10) in~\cite{Olfati2009DKF}), where the optimal predicted information covariance matrix, $\bar{\mathbf{S}}_i^*$, is fused with all the measurement information covariance matrices, $\mathbf{Y}_i$, because the measurements are uncorrelated and $\bar{\mathbf{S}}_i^*$ is consistent (proven in Proposition~\ref{proposition:consistency}). The estimated state is obtained as follows:
\begin{equation}\label{eq:CODKF_estimation}
    \hat{\mathbf{x}}_i^* = \bar{\mathbf{x}}_i^* + \hat{\mathbf{P}}_{i}^*(\mathbf{y}_i - \mathbf{Y}_i\bar{\mathbf{x}}_i^*).
\end{equation}
It is noteworthy to see that Eq.~\eqref{eq:CODKF_estimation} is similar to 
\begin{equation}
    \label{eq:KF4}
    \hat{\mathbf{x}}_i =
        \bar{\mathbf{x}}_i + \hat{\mathbf{P}}_{i}^*\left(\mathbf{y}_i - \mathbf{Y}_i\bar{\mathbf{x}}_i\right) +\gamma \hat{\mathbf{P}}_{i}^* \sum_{j\in \mathcal{N}_i} (\bar{\mathbf{x}}_j - \bar{\mathbf{x}}_i),
\end{equation}
which comes from the consensus-based DKF in~\cite{Olfati2007DKF}. However, in ECO-DKF the consensus is implicit in the optimisation, overcoming the dependence for stability on the parameter $\gamma$. 

We briefly discuss the communication and computational burden. Regarding communication, the size of each message is constant with the number of nodes, so the proposal is scalable. Besides, nodes do not need any global knowledge of the topology nor sensor models of neighbours. Regarding the computational cost, the bottleneck is~\eqref{eq:outerLJellipse}. However, the current hardware can solve large instances of this optimisation problem in real-time~\cite{Yang2020Teaser}.

\begin{algorithm}
\caption{ECO-DKF in node $i$}\label{al:ECO-DKF}
\begin{algorithmic}[1]
\STATE Initialisation: $\hat{\mathbf{P}}_i^*(0) = \mathbf{P}_0$, $\hat{\mathbf{x}}_i^*(0) = \mathbf{x}_0$, $\bm{\lambda}_i^*(0)=\bm{\lambda}_0$
\WHILE{True}
    \STATE Prediction:
            \begin{itemize}
               \item[] $\bar{\mathbf{P}}_i(k) = \mathbf{A}\hat{\mathbf{P}}_i^*(k-1)\mathbf{A}^T + \mathbf{Q}$
               \item[] \:$\bar{\mathbf{x}}_i(k) = \mathbf{A}\hat{\mathbf{x}}_i^*(k-1)$
           \end{itemize}
    \STATE Measurement: get $\mathbf{z}_i(k)$
    \STATE Broadcasting:
    \STATE \hspace{0.7cm}\textbf{if} $g_i(\bm{\lambda}^*_i(k-1), k) == 1$ \textbf{then}
    \STATE \hspace{1.4cm}Send $\{ \mathbf{U}_i(k), \mathbf{u}_i(k), \bar{\mathbf{S}}_i(k), \bar{\mathbf{s}}_i(k)  \}$
    \STATE \hspace{0.7cm}\textbf{end if}
    \STATE \hspace{0.7cm}Receive $\{ \mathbf{U}_j(k), \mathbf{u}_j(k), \bar{\mathbf{S}}_j(k), \bar{\mathbf{s}}_j(k)  \}$ $\forall j \in \mathcal{N}_i(k)$
    \STATE Aggregation of measurements' data:
           \begin{itemize}
               \item[] $\mathbf{Y}_i(k) = \sum_{j\in \mathcal{J}_i(k)} \mathbf{U}_j(k)$ 
               \item[] $\mathbf{y}_i(k) \:= \sum_{j\in \mathcal{J}_i(k)}  \mathbf{u}_j(k)$ 
           \end{itemize}
    \STATE Aggregation of predictions' data:
            \begin{itemize}
               \item[] \:$\bar{\mathbf{S}}_i^*(k)$,$\bm{\lambda}^*_i(k)$ $\leftarrow$ Solution of~\eqref{eq:outerLJellipse}
               \item[] \:\:\:\:\:\:$\bar{\mathbf{P}}_{i}^*(k) = (\bar{\mathbf{S}}_i^*(k))^{-1}$
               \item[] \:\:\:\:\:\:$\bar{\mathbf{x}}_{i}^*(k) \:= \bar{\mathbf{P}}^*_{i}(k) \sum_{j\in \mathcal{J}_i(k)} \lambda_{ij}^*(k)\bar{\mathbf{s}}_j(k)$
           \end{itemize}
    \STATE Correction:
            \begin{itemize}
               \item[] $\hat{\mathbf{P}}_i^*(k) = (\bar{\mathbf{S}}^*_i(k) + \mathbf{Y}_i(k))^{-1}$
               \item[] \:\:$\hat{\mathbf{x}}_i^*(k) = \bar{\mathbf{x}}_{i}^*(k) + \hat{\mathbf{P}}_i^*(k)(\mathbf{y}_i(k) - \mathbf{Y}_i(k)\bar{\mathbf{x}}_{i}^*(k))$
           \end{itemize}
\ENDWHILE 
\end{algorithmic}
\end{algorithm}


\section{Certifiable Covariance Bounding}\label{sec:solution}
To achieve the optimality pursued in Problem~\ref{problem:ECO-DKF} it is necessary to deal with the optimal aggregation of neighbouring predictions under unknown correlations. This fusion can be described as finding the minimum volume ellipsoid containing the intersection of the $p$ ellipsoids~\cite{Boyd1994LMI} formed by the matrices $\bar{\mathbf{P}}_j$ in the set $\mathcal{J}_i(k)$. 
Thus, we first define the concept of intersection of ellipsoids.
\begin{definition}\label{definition:intersection}
The intersection of $p$ ellipsoids, at node $i$, is the convex set 
$\mathcal{F}_i := \varepsilon_1^i \cap \hdots \cap \varepsilon_j^i \cap \hdots \cap \varepsilon_p^i$.
\end{definition}

From Definitions~\ref{definition:ellipsoid} and~\ref{definition:intersection}, volume minimisation can be transformed into a maximisation over the information matrices. In particular since the expected square error is equivalent to the trace of $\hat{\mathbf{P}}_i^*,$ the objective of ECO-DKF is to optimise this metric.
\begin{problem}[\textbf{Adapted from~\cite{Boyd1994LMI}}]\label{problem:intersection}
Find $(\varepsilon^i_i)^{*}$ such that it contains $\mathcal{F}_i$ and \emph{Tr}$(\bar{\mathbf{S}}_i^*)$ is maximised.
\end{problem}
Problem~\ref{problem:intersection} is solved by the following optimisation problem~\cite{Henrion2001LMI}: 
\begin{subequations}\label{eq:traceProblem}
\begin{alignat}{2}
\underset{\mathbf{X}}{\max} & \:\:\:\:\:\:\:\hbox{Tr}(\mathbf{X}) \label{eq:optProb4}
\\
s.t. & \:\:\:\:\:\:\: \mathbf{X}    \succeq \mathbf{0}, \hbox{ }   \hbox{Tr}(\mathbf{X} \bar{\mathbf{S}}_j)  \leq 1 \:\:\: \forall j \in \mathcal{J}_i(k),\label{eq:constraint41}
\\
&   \:\:\:\:\:\:\:   \hbox{rank}(\mathbf{X})  = 1. \label{eq:constraint43}
\end{alignat}
\end{subequations}
This problem is NP-hard due to the non-convex constraint~\eqref{eq:constraint43}. Thus, it is necessary to find a convex relaxation to make the optimisation tractable. 
The simplest solution is to drop the non-convex constraint,
\begin{subequations}\label{eq:traceRelax}
\begin{alignat}{2}
\underset{\mathbf{X}}{\max} & \:\:\:\:\:\:\: \hbox{Tr}(\mathbf{X}) \label{eq:optProb5}
\\
s.t. & \:\:\:\:\:\:\: \mathbf{X}    \succeq \mathbf{0}, \hbox{ }    \hbox{Tr}(\mathbf{X}\bar{\mathbf{S}}_j)  \leq 1 \:\:\: \forall j \in \mathcal{J}_i(k),\label{eq:constraint51}
\end{alignat}
\end{subequations}
However, it is not possible to use~\eqref{eq:traceRelax} in ECO-DKF because it does not provide an estimate nor an error covariance matrix, but $\mathbf{X}$. Instead, we propose to use the outer Löwner-John (LJ) method~\cite{John2014Ellipsoid}, which leads to the relaxed optimisation problem in~\eqref{eq:outerLJellipse}. 
The input of the outer LJ method is the $|\mathcal{J}_i(k)|$ ellipsoids described by matrices $\bar{\mathbf{S}}_j$ $\forall j \in \mathcal{J}_i(k)$. Then, it computes the smallest ellipsoid that includes $\mathcal{F}_i(k)$. The intersection always exist since the ellipsoids are all centered at $\mathbf{0}$~\cite{Henrion2001LMI}. From~\eqref{eq:constraint31}, $\bar{\mathbf{S}}_i^*$ is upper-bounded by the weighted sum of matrices $\bar{\mathbf{S}}_j$ $\forall j \in \mathcal{J}_i(k)$. Therefore, $\bar{\mathbf{S}}_i^*$ depends on the topology and the target system dynamics, for they determine the set of predicted error covariance matrices considered in the optimisation (Eq.~\eqref{eq:constraint31}) and the value of each prediction (Eqs.~\eqref{eq:KF5}-\eqref{eq:KF6}) respectively. 
Problem~\eqref{eq:outerLJellipse} outputs an estimate and error covariance matrix.
Nevertheless, problem~\eqref{eq:traceRelax} is important because it enables the certification of optimality. This is formally stated in the next Proposition, whose proof can be found in~\cite{Sebastian2021CDC}:
\begin{proposition}\label{proposition:certifiability}
Let $\mathbf{X}^*$ be the solution of~\eqref{eq:traceRelax}.
Define
\begin{equation}
    \mathbb{C}_i := \emph{rank}\left(\mathbf{X}^*\right) \hbox{ and } 
    \rho_i := \emph{Tr}(\mathbf{X}^*) \vartheta(\bar{\mathbf{S}}_i^*)\in [0,1],
\end{equation}
where $\vartheta(\bar{\mathbf{S}}_i^*)$ denotes the minimum eigenvalue of $\bar{\mathbf{S}}_i^*$, obtained by solving~\eqref{eq:outerLJellipse}.
If $\mathbb{C}_i=\rho_i=1$, then the solution of~\eqref{eq:outerLJellipse} is the optimum of the original non-relaxed problem~\eqref{eq:traceProblem}.
\end{proposition}

The proposition gives a procedure where each node can check, locally and in real-time, the performance of its optimisation process and assess it by finding the optimal value for Problem~\ref{problem:intersection}. This is of key importance not only to verify the optimality of the estimation, but also the certificate can be used in other parts of the control pipeline to, e.g., assure optimality of the controller. In our case, we will leverage this process to design an ET communication rule.


\section{Stability and Optimality}\label{sec:stability}
In this section we analyse the stability and optimality of ECO-DKF. For clarity, we assume that all nodes broadcast their information to neighbours at all instants, such that the underlying communication graph is undirected and connected. In the next section we will relax the assumption.

Define $\eta_i = \hat{\mathbf{x}}_i-{\mathbf{x}}$,
\begin{equation*}
    \begin{aligned}
        \eta =& [\eta_1^T, \hdots, \eta_{N}^T]^T,  & \mathbf{H} =& \hbox{block-diag}(\mathbf{H}_1, \hdots, \mathbf{H}_{N})
        \\
        \mathbf{z} =& [\mathbf{z}_1^T, \hdots, \mathbf{z}_N^T]^T, &\mathbf{R} =& \hbox{block-diag}(\mathbf{R}_1, \hdots, \mathbf{R}_{N}),
        \\
        \mathbf{y} =& [\mathbf{y}_1^T, \hdots, \mathbf{y}_N^T]^T,&\mathbf{U} =& \hbox{block-diag}(\mathbf{U}_1, \hdots, \mathbf{U}_{N}),
        \\
        \mathbf{\mathcal{A}} =& \mathbf{I} \otimes \mathbf{A}, &\hat{\mathbf{P}}^* =& \hbox{block-diag}(\hat{\mathbf{P}}^*_1, \hdots, \hat{\mathbf{P}}^*_{N}),
        \\
        \mathbf{\mathcal{Q}} =& \mathbf{I} \otimes \mathbf{Q},&\bar{\mathbf{P}}^* =& \hbox{block-diag}(\bar{\mathbf{P}}_{1}^*, \hdots, \bar{\mathbf{P}}_{N}^*),
    \end{aligned}
\end{equation*}
and
\begin{equation}\label{eq:nw}
\begin{aligned}
    \mathbf{Y} \overset{Eq.~\eqref{eq:KF1}}{=} \mathbf{N}_w \mathbf{U} \hbox{ , where } & 
    \left. \begin{array}{ll}
         \lfloor \mathbf{N}_w \rfloor_{ij} = \mathbf{I}  \hbox{ if } j\in \mathcal{J}_i  \\
         \lfloor \mathbf{N}_w \rfloor_{ij} = 0*\mathbf{I} \hbox{ otherwise}
    \end{array}\right..
\end{aligned}
\end{equation}
To demonstrate the global asymptotic stability of ECO-DKF, we adapt and prove two Lemmas from~\cite{Olfati2009DKF}, \blue{extending them from the discrete-time centralized Kalman filter setting to our ECO-DKF setting}. The first builds some useful matrices. 
\begin{lemma}[\textbf{Adapted from Lemma 2 in~\cite{Olfati2009DKF}}]\label{lemma:matrices}
Given Eqs.~\eqref{eq:KF3}-\eqref{eq:CODKF_estimation}, the following holds:
\begin{enumerate}
    \item $\mathbf{F} = \mathbf{I} - \hat{\mathbf{P}}^*\mathbf{Y}=\hat{\mathbf{P}}^*(\bar{\mathbf{P}}^*)^{-1}$.
    \item $\hat{\mathbf{P}}^* = \mathbf{F}\mathbf{G}\mathbf{F}^T$ with $\mathbf{G} = \mathbf{\mathcal{A}}\hat{\mathbf{P}}^*\mathbf{\mathcal{A}}^T + \mathbf{\mathcal{Q}} + \mathbf{T}\mathbf{R}^{-1}\mathbf{T}^T$ and $\mathbf{T} = \bar{\mathbf{P}}^* \mathbf{N}_w \mathbf{H}^T$.
\end{enumerate}
\end{lemma}
\begin{proof}
Let $\mathbf{F} = \mathbf{I} - \hat{\mathbf{P}}^*\mathbf{Y}$, by Eq.~\eqref{eq:KF3} we can write $\hat{\mathbf{P}}^* = ((\bar{\mathbf{P}}^*)^{-1} + \mathbf{Y})^{-1}$. Then,
\begin{equation}
   \hat{\mathbf{P}}^*((\bar{\mathbf{P}}^*)^{-1} + \mathbf{Y}) = \mathbf{I} \Rightarrow \mathbf{I} - \hat{\mathbf{P}}^*\mathbf{Y}=\hat{\mathbf{P}}^*(\bar{\mathbf{P}}^*)^{-1}
\end{equation}
and the first statement is proved. Regarding the second statement, notice that Eq.~\eqref{eq:CODKF_estimation} can be stacked as follows
\begin{equation}\label{eq:CODKF_estimation_stacked}
    \hat{\mathbf{x}}^* = \bar{\mathbf{x}}^* + \hat{\mathbf{P}}^*(\mathbf{y} - \mathbf{Y}\bar{\mathbf{x}}^*).
\end{equation}
Then, taking into account that $\mathbf{y}=\mathbf{N}_w\mathbf{H}^T\mathbf{R}^{-1}\mathbf{z}$ and $\mathbf{Y}=\mathbf{N}_w\mathbf{U} = \mathbf{N}_w\mathbf{H}^T\mathbf{R}^{-1}\mathbf{H}$, Eq.~\eqref{eq:CODKF_estimation_stacked} is rewritten from its information form as
\begin{equation}\label{eq:CODKF_estimation_standard}
    \hat{\mathbf{x}}^* = \bar{\mathbf{x}}^* + \hat{\mathbf{P}}^*\mathbf{N}_w\mathbf{H}^T\mathbf{R}^{-1}(\mathbf{z} - \mathbf{H}\bar{\mathbf{x}}^*) = \bar{\mathbf{x}}^* + \mathbf{K}(\mathbf{z} - \mathbf{H}\bar{\mathbf{x}}^*) 
\end{equation}
with $\mathbf{K}=\hat{\mathbf{P}}^*\mathbf{N}_w\mathbf{H}^T\mathbf{R}^{-1}$ the Kalman Gain of the standard KF formulation. Let now consider the update of $\hat{\mathbf{P}}^*$, where by update, denoted as $(\hat{\mathbf{P}}^*)^{+}$, we refer to the operations that compute $\hat{\mathbf{P}}^*(k)$ from $\hat{\mathbf{P}}^*(k-1)$ and $\bar{\mathbf{P}}^*(k-1)$
\begin{equation}\label{eq:standard}
    (\hat{\mathbf{P}}^*)^{+} = \mathbf{F}\bar{\mathbf{P}}^+\mathbf{F}^T + \mathbf{K}\mathbf{R}\mathbf{K}^T,  
\end{equation}
which comes from the standard KF~\cite{Olfati2009DKF}. By substituting $\mathbf{K}$ we obtain
\begin{equation}\label{eq:standard2}
    (\hat{\mathbf{P}}^*)^{+} = \mathbf{F}(\mathbf{\mathcal{A}}\hat{\mathbf{P}}^*\mathbf{\mathcal{A}}^T + \mathbf{\mathcal{Q}})\mathbf{F}^T \kern -0.2cm + \hat{\mathbf{P}}^*\mathbf{N}_w\mathbf{H}^T\mathbf{R}^{-1}\mathbf{H}\mathbf{N}_w^T(\hat{\mathbf{P}}^*)^T.   
\end{equation}
Given statement 1 of the Lemma,
\begin{equation}\label{eq:standard3}
    (\hat{\mathbf{P}}^*)^{+} \kern -0.2cm = \mathbf{F}(\mathbf{\mathcal{A}}\hat{\mathbf{P}}^*\mathbf{\mathcal{A}}^T + \mathbf{\mathcal{Q}} + \bar{\mathbf{P}}^*\mathbf{N}_w\mathbf{H}^T\mathbf{R}^{-1}\mathbf{H}\mathbf{N}_w^T(\bar{\mathbf{P}}^*)^T)\mathbf{F}^T,   
\end{equation}
and statement 2 is proved.
\end{proof}
The next Lemma proves global asymptotic stability for general dynamics, used later to prove ECO-DKF global asymptotic stability.
\begin{lemma}[\textbf{Adapted from Lemma 3 in~\cite{Olfati2009DKF}}]\label{lemma:stability}
Suppose that the error dynamics without noise are $\eta^+ = \mathbf{F}\mathbf{\mathcal{A}}\eta$, with $\mathbf{F}$ as in Lemma~\ref{lemma:matrices}. Then, the error dynamics is globally asymptotically stable system with a Lyapunov function $V(\eta) = \eta^T (\hat{\mathbf{P}}^*)^{-1} \eta$, provided that $g_i(\bm{\lambda}^*_i)=1$ for all $i,k$ and $G(k)$ is undirected and connected for all $k$. 
\end{lemma}
\begin{proof}
Given $V(\eta) = \eta^T (\hat{\mathbf{P}}^*)^{-1} \eta$ as Lyapunov function candidate,
\begin{equation}
\begin{aligned}
    \delta V =& (\eta^+)^T ((\hat{\mathbf{P}}^*)^+)^{-1} \eta^+ - \eta^T (\hat{\mathbf{P}}^*)^{-1} \eta =\\
              & \eta^T\left.(\mathbf{\mathcal{A}}^T\mathbf{F}^T ((\hat{\mathbf{P}}^*)^+)^{-1} \mathbf{F} \mathbf{\mathcal{A}} - (\hat{\mathbf{P}}^*)^{-1})\right. \eta = \\
              & \eta^T\left.(\mathbf{\mathcal{A}}^T\mathbf{G}^{-1} \mathbf{\mathcal{A}} - (\hat{\mathbf{P}}^*)^{-1})\right. \eta = \\
              & -\eta^T\left.((\hat{\mathbf{P}}^*)^{-1} - \mathbf{\mathcal{A}}^T(\mathbf{\mathcal{A}}\hat{\mathbf{P}}^*\mathbf{\mathcal{A}}^T + \mathbf{W})^{-1} \mathbf{\mathcal{A}})\right. \eta = \\
              & -\eta^T \Lambda \eta,
\end{aligned}
\end{equation}
with $\mathbf{W} = \mathbf{\mathcal{Q}} + \mathbf{T}\mathbf{R}^{-1}\mathbf{T}^T \succ \mathbf{0}$ and $\Lambda = (\hat{\mathbf{P}}^*)^{-1} - \mathbf{\mathcal{A}}^T(\mathbf{\mathcal{A}}\hat{\mathbf{P}}^*\mathbf{\mathcal{A}}^T + \mathbf{W})^{-1}\mathbf{\mathcal{A}}$. The rest of the proof directly follows from Lemma 3 of~\cite{Olfati2009DKF}, showing that $\Lambda \succ \mathbf{0}$.
\end{proof}
Finally, to prove stability of ECO-DKF we employ a Lemma, whose proof can be found in~\cite{Sebastian2021CDC}, which allows to rewrite ECO-DKF filter equations\blue{ and show that the optimum of optimization problem~\eqref{eq:outerLJellipse} is such that inequality~\eqref{eq:constraint31} is an equality.} 
\begin{lemma}
\label{lemma:convergence}
Given optimisation problem~\eqref{eq:outerLJellipse}, \blue{$\bar{\mathbf{S}}^*_i = \sum_{j \in \mathcal{J}_i} \lambda_{ij}^* \bar{\mathbf{S}}_j$}, for all $i$. 
\end{lemma}

Lemma~\ref{lemma:convergence} says that the optimisation in~\eqref{eq:outerLJellipse} becomes a standard discrete-time consensus protocol tuned to optimise the trace of the final consensus value of each node. Interestingly, the optimisation will implicitly assign largest $\lambda_{ij}^*$ to nodes equipped with better sensors and estimates, which is a positive side effect.

These ingredients lead to the stability of the filter.
\begin{theorem}\label{theorem:stability}
Algorithm~\ref{al:ECO-DKF} is a globally asymptotically stable estimator if $g_i(\bm{\lambda}^*_i)=1$ for all $i,k$.
\end{theorem}
\begin{proof}
First, rewrite Eq.~\eqref{eq:CODKF_estimation} like Eq.~\eqref{eq:KF4}. The aggregated prediction $\bar{\mathbf{x}}^*_{i}$ in Eq.~\eqref{eq:CODKF_estimation} is 
\begin{equation}\label{eq:xfusiones}
     \bar{\mathbf{x}}_{i}^* = \bar{\mathbf{P}}_{i}^* \sum_{j\in \mathcal{J}_i} \lambda_{ij}^*\bar{\mathbf{P}}_j^{-1}\bar{\mathbf{x}}_{j}.
\end{equation}
This expression can be rewritten as
\begin{equation}\label{eq:xfusiones3}
     \bar{\mathbf{x}}_{i}^* = \bar{\mathbf{P}}_{i}^*\sum_{j\in \mathcal{J}_i} \lambda_{ij}^*\bar{\mathbf{P}}_j^{-1}  \bar{\mathbf{x}}_{i}  +
     \bar{\mathbf{P}}^*_{i} \sum_{j\in \mathcal{N}_i} \lambda_{ij}^*\bar{\mathbf{P}}_j^{-1}(\bar{\mathbf{x}}_{j}-\bar{\mathbf{x}}_{i}).
\end{equation}
From Lemma~\ref{lemma:convergence}, $\bar{\mathbf{P}}_{i}^*\sum_{j\in \mathcal{J}_i} \lambda_{ij}^*\bar{\mathbf{P}}_j^{-1} \kern -0.1cm= \mathbf{I}$. Using~\eqref{eq:xfusiones3} in~\eqref{eq:CODKF_estimation} gives
\begin{equation}\label{eq:rewrite_sta}
\begin{aligned}
     \hat{\mathbf{x}}_{i}^* \kern -0.1cm= \kern -0.1cm\bar{\mathbf{x}}_{i} \kern -0.1cm + \kern -0.1cm\hat{\mathbf{P}}^*_i(\mathbf{y}_i \kern -0.1cm - \kern -0.1cm\mathbf{Y}_i\bar{\mathbf{x}}_{i})  \kern -0.1cm+ \kern -0.1cm
     (\mathbf{I}\kern -0.1cm-\kern -0.1cm\hat{\mathbf{P}}^*_i\mathbf{Y}_i)\bar{\mathbf{P}}_{i}^*
     \kern -0.1cm\sum_{j\in \mathcal{N}_i} \kern -0.1cm\lambda_{ij}^*\bar{\mathbf{P}}_j^{-1}(\bar{\mathbf{x}}_{j}\kern -0.1cm-\kern -0.1cm\bar{\mathbf{x}}_{i}). 
\end{aligned}
\end{equation}
Eq.~\eqref{eq:rewrite_sta} is equivalent to Eq.~\eqref{eq:KF4} but weighting each term with the result of optimising~\eqref{eq:outerLJellipse}.
Then, the noiseless dynamics of $\eta_i$ is
\begin{equation}\label{eq:error_dynamics}
\begin{aligned}
     \eta_{i}^+ \kern -0.1cm=\kern -0.1cm \mathbf{A}\eta_{i} \kern -0.1cm- \kern -0.1cm\hat{\mathbf{P}}^*_i\mathbf{Y}_i\mathbf{A}\eta_{i}  \kern -0.1cm+ \kern -0.1cm
     (\mathbf{I}\kern -0.1cm-\kern -0.1cm\hat{\mathbf{P}}^*_i\mathbf{Y}_i)\bar{\mathbf{P}}_{i}^*\kern -0.1cm\sum_{j\in \mathcal{N}_i}\kern -0.1cm \lambda_{ij}^*\bar{\mathbf{P}}_j^{-1}\mathbf{A}(\eta_{j}\kern -0.1cm-\kern -0.1cm\eta_{i})
\end{aligned}
\end{equation}
which can also be written in compact form for the whole network as 
\begin{equation}\label{eq:error_dynamics_overall_re}
 \eta^+ = \mathbf{F}\left.\mathbf{L}_w\right. \mathbf{\mathcal{A}}\eta.
\end{equation}
where $(\cdot)^+$ is the update operator defined as in Lemma~\ref{lemma:stability} and the proof of Lemma~\ref{lemma:matrices}, $\mathbf{F}$ is defined as in Lemma~\ref{lemma:matrices}, and $\mathbf{L}_w$ is such that $\lfloor \mathbf{L}_w \rfloor_{ij} = \bar{\mathbf{P}}_{i}^*\lambda_{ij}^*\bar{\mathbf{P}}_j^{-1}$ for all $j\in \mathcal{J}_i$ and $\mathbf{0}$ otherwise. The last expression is similar to the error dynamics proved as globally asymptotically stable in Lemma~\ref{lemma:stability}, but with $\mathbf{L}_w$ in between. In the proof of Lemma~\ref{lemma:stability} it is shown that 
\begin{equation}
    \Lambda = (\hat{\mathbf{P}}^*)^{-1} - \mathbf{\mathcal{A}}^T\mathbf{G}^{-1} \mathbf{\mathcal{A}} \succ \mathbf{0}.
\end{equation}
Instead, we have that $\Lambda^{'} = (\hat{\mathbf{P}}^*)^{-1} - \mathbf{\mathcal{A}}^T\mathbf{L}_w^T\mathbf{G}^{-1}\mathbf{L}_w\mathbf{\mathcal{A}}$. If 
\begin{equation}\label{eq:statement1}
 \mathbf{\mathcal{A}}^T\mathbf{G}^{-1} \mathbf{\mathcal{A}} \succeq \mathbf{\mathcal{A}}^T\mathbf{L}_w^T\mathbf{G}^{-1}\mathbf{L}_w\mathbf{\mathcal{A}} \quad \Rightarrow \quad \mathbf{I} \succeq \mathbf{L}_w, 
\end{equation}
then $\Lambda^{'} \succ \mathbf{0}$ and global asymptotic stability is proved. Given Lemma~\ref{lemma:convergence}, $\mathbf{L}_w$ is a row-stochastic matrix. Therefore, by linear algebra results, the absolute value of any eigenvalue of $\mathbf{L}_w$ is less than or equal to $1$. This means that the eigenvalues of matrix $\mathbf{L}_w - \mathbf{I}$ are all negative or equal to $0$ and  
\begin{equation}\label{eq:statement3}
 \mathbf{0} \succeq \mathbf{L}_w - \mathbf{I} \Rightarrow \mathbf{I} \succeq \mathbf{L}_w.
\end{equation}
Thus,~\eqref{eq:statement1} holds and $\Lambda^{'} \succ \mathbf{0}$, concluding the proof.
\end{proof}

Notice that stability holds independently on certification. 
The next step is to demonstrate optimality under positive certification. 

\begin{theorem}\label{theorem:optimality}
 Assume that the solution of~\eqref{eq:outerLJellipse} is certified as optimal. Then, Algorithm~\ref{al:ECO-DKF} is optimal in the sense of Problem~\ref{problem:ECO-DKF} and~\eqref{eq:outerLJellipse} provides the optimal consensus gain of Eq.~\eqref{eq:KF4}.
\end{theorem}
\begin{proof}
\blue{We reformulate Problem~\ref{problem:ECO-DKF} as an optimal Bayesian estimation problem over the network $\mathcal{G}$. Let $Z(k) = \{\mathbf{z}(0), \hdots, \mathbf{z}(k) \}$, where $\mathbf{z} = [\mathbf{z}^T_1, \hdots, \mathbf{z}^T_N]^T$. The posterior that minimises the MSE is given by $\hat{\mathbf{x}}(k) = \text{E}[\mathbf{x}(k)|Z(k)] = \int \mathbf{x}(k)P(\mathbf{x}(k)|Z(k))d\mathbf{x}(k)$. To obtain $P(\mathbf{x}(k)|Z(k))$ from $P(\mathbf{x}(k-1)|Z(k-1))$, the Bayesian prediction and correction equations are:
\begin{equation}
    \begin{aligned}
    &\kern -0.3cmP(\mathbf{x}(k)|Z(k\kern -0.1cm-\kern -0.1cm1)) = \kern -0.2cm \int \kern -0.2cm P(\mathbf{x}(k)|\mathbf{x}(k\kern -0.1cm-\kern -0.1cm1))P(\mathbf{x}(k\kern -0.1cm-\kern -0.1cm1)|Z(k\kern -0.1cm-\kern -0.1cm1))d\mathbf{x}(k\kern -0.1cm-\kern -0.1cm1)    
    \\
    &\kern -0.3cmP(\mathbf{x}(k)|Z(k)) =  \frac{P(\mathbf{z}(k) | \mathbf{x}(k))P(\mathbf{x}(k)| Z(k\kern -0.1cm-\kern -0.1cm1))}{\int P(\mathbf{z}(k)|\mathbf{x}(k))P(\mathbf{x}(k)|Z(k\kern -0.1cm-\kern -0.1cm1))d\mathbf{x}(k)}        
    \end{aligned}
\end{equation}
Under unknown correlations, $P(\mathbf{x}(k)|Z(k-1)) = \mathcal{N}(\bar{\mathbf{x}}(k), \bar{\mathbf{P}}(k)):$ 
\begin{equation}\label{eq:proof_opt1}
    \bar{\mathbf{x}}(k) =\mathcal{A}\hat{\mathbf{x}}^*(k-1) 
    \quad \text{ and }  \quad
    \bar{\mathbf{P}}(k) =  \mathcal{A}\hat{\mathbf{P}}^*(k-1)\mathcal{A}^T + \mathcal{Q},
\end{equation}
with $\hat{\mathbf{P}}^*(k\kern -0.1cm-\kern -0.1cm1)\kern -0.1cm=\kern -0.1cm\text{block-diag}(\hat{\mathbf{P}}_1^*(k\kern -0.1cm-\kern -0.1cm1), \hdots, \hat{\mathbf{P}}_N^*(k\kern -0.1cm-\kern -0.1cm1))$. This is the prediction step. Regarding the correction step, first, we have that \mbox{$\mathbf{y} = [\mathbf{y}_1^T, \hdots, \mathbf{y}_N^T]^T$}, \mbox{ $\mathcal{R} = \text{block-diag}(\mathbf{Y}_{1}^{-1}, \hdots, \mathbf{Y}_{N}^{-1})$}, and  
\mbox{$P(\mathbf{z}(k)|\mathbf{x}(k)) = \mathcal{N}(\mathcal{H}\mathbf{x}(k), \mathcal{R}),$}
where $\mathcal{H}$ is the block diagonal matrix such that $\mathbf{Y}^{-1}\mathbf{y}(k) = \mathcal{H}\mathbf{x}(k)$. On the other hand, \mbox{$ P(\mathbf{x}(k)|Z(k)) = \mathcal{N}(\bar{\mathbf{x}}^*(k), \bar{\mathbf{P}}^*(k))$}, where $\bar{\mathbf{x}}^*(k)$ and $\bar{\mathbf{P}}^*(k)$ are the fused predictions at each node. In particular, given the restrictions in Problem~\ref{problem:ECO-DKF}, the fused prediction that ensures consistency (constraint 3) under unknown correlations (constraint 1) and one-hop communications (constraint 2) is given by the outer Löwner-John (LJ) method~\cite{John2014Ellipsoid}. Among the different measures of ellipsoid size, $\hbox{Tr}(\bar{\mathbf{S}}^{-1})$ in problem~\eqref{eq:outerLJellipse} minimises the MSE because $\text{MSE} = \hbox{E}\left[\frac{1}{N}\sum_{i=1}^N||\hat{\mathbf{x}}^*_i-\mathbf{x}||^2\right] =  \hbox{Tr}(\hbox{block-diag}(\hat{\mathbf{P}}^*_1, \hdots, \hat{\mathbf{P}}^*_N)) = \hbox{Tr}(\hat{\mathbf{P}}^*) = \hbox{Tr}((\hat{\mathbf{S}}^*)^{-1}).$ By assumption of the theorem, the solution of~\eqref{eq:outerLJellipse} is optimal, so the associated unbiased optimal fused prediction at each node is \mbox{$\bar{\mathbf{x}}_{i}^* = \bar{\mathbf{P}}_{i}^* \sum_{j\in \mathcal{J}_i} \lambda_{ij}^*\bar{\mathbf{P}}_j^{-1}\bar{\mathbf{x}}_{j},$} which is equal to Eq.~\eqref{eq:KF_o1}. 

Using the fundamental Gaussian identity (Appendix D in \cite{Mahler2007Statistical}),  $\mathcal{N}(\mathcal{H}\mathbf{x}(k), \mathcal{R})\mathcal{N}(\bar{\mathbf{x}}^*(k), \bar{\mathbf{P}}^*(k)) = \mathcal{N}(\mathcal{H}\bar{\mathbf{x}}^*(k), \mathcal{R} + \mathcal{H}\bar{\mathbf{P}}^*(k)\mathcal{H}^T)\mathcal{N}(\bar{\mathbf{x}}^*(k), ((\bar{\mathbf{P}}^*(k))^{-1} + \mathcal{H}{\mathcal{R}}^{-1}\mathcal{H}^T)^{-1}).$ Besides, the integral $\int \mathcal{N}(\mathcal{H}\mathbf{x}(k), \mathcal{R})\mathcal{N}(\bar{\mathbf{x}}^*(k), \bar{\mathbf{P}}^*(k))d\mathbf{x}(k) = \mathcal{N}(\mathcal{H}\bar{\mathbf{x}}^*(k), \mathcal{R} + \mathcal{H}\bar{\mathbf{P}}^*(k)\mathcal{H}^T),$ and, therefore, $P(\mathbf{x}(k)|Z(k)) = \mathcal{N}(\hat{\mathbf{x}}^*(k), \hat{\mathbf{P}}^*(k))$. Finally, by definition, $(\hat{\mathbf{P}}^*(k))^{-1}\hat{\mathbf{x}}^*(k) = (\bar{\mathbf{P}}^*(k))^{-1}\bar{\mathbf{x}}^*(k) + \mathcal{H}\mathcal{R}^{-1}\mathbf{z}(k),$ and $\hat{\mathbf{P}}^*(k)(\hat{\mathbf{P}}^*(k))^{-1}\hat{\mathbf{x}}^*(k) = \hat{\mathbf{x}}^*(k) = \bar{\mathbf{x}}^*(k) + \mathbf{K}(k)(\mathbf{z}(k)  - \mathcal{H}\bar{\mathbf{x}}^*(k)).$ Using Lemma~\ref{lemma:matrices}, the update/correction operations are written in information form:
\begin{equation}\label{eq:proof_opt2}
    \begin{aligned}
        &\hat{\mathbf{P}}^*(k) = ((\bar{\mathbf{P}}^*(k))^{-1} + \mathcal{R}^{-1})^{-1}
        \\
        &
        \hat{\mathbf{x}}^*(k) = \bar{\mathbf{x}}^*(k) + \bar{\mathbf{P}}^*(k)(\mathbf{y}(k)  - \mathcal{R}^{-1})\bar{\mathbf{x}}^*(k)).
    \end{aligned}
\end{equation}
Equations~\eqref{eq:proof_opt1} and \eqref{eq:proof_opt2}, together with optimisation problem \eqref{eq:outerLJellipse}, are the ECO-DKF filter in Algorithm~\ref{al:ECO-DKF}. Besides, the expression in Eq.~\eqref{eq:rewrite_sta} can be reformulated as
\begin{equation}\label{eq:rewrite}
     \hat{\mathbf{x}}_{i}^* = \bar{\mathbf{x}}_{i} + \hat{\mathbf{P}}^*_i(\mathbf{y}_i - \mathbf{Y}_i\bar{\mathbf{x}}_{i})  + 
     \gamma\hat{\mathbf{P}}^*_i\textstyle\sum_{j\in \mathcal{N}_i} \lambda_{ij}^*\bar{\mathbf{P}}_j^{-1}(\bar{\mathbf{x}}_{j}-\bar{\mathbf{x}}_{i})
\end{equation}
with $\gamma = (\mathbf{I}-\hat{\mathbf{P}}^*_i\mathbf{Y}_i)\bar{\mathbf{P}}_{i}^*(\hat{\mathbf{P}}^*_i)^{-1}$. The latter is equivalent to Eq.~\eqref{eq:KF4} but weighting each term with the result of optimising~\eqref{eq:outerLJellipse}.

Therefore, Algorithm~\ref{al:ECO-DKF} is optimal in the sense of Problem~\ref{al:ECO-DKF} and~\eqref{eq:outerLJellipse} provides the optimal consensus gain of Eq.~\eqref{eq:KF4}.}
\end{proof}


\section{Event-triggered scheme}\label{sec:event_triggered}
In this section we describe the proposed ET general scheme and  we study its theoretical properties to show that the stability and optimality results are preserved under PJC of the network.

\subsection{Description of the ET scheme}\label{subsec:descriptionET}

Under an ET scheme, a node should communicate when it possesses relevant information for the network. Our ET rule leverages the output of~\eqref{eq:outerLJellipse} to decide the triggering. Remember that~\eqref{eq:outerLJellipse} finds the optimal combination of estimates under positive certification. In particular, $\bm{\lambda}_i^*$ weights how important is each prediction in the optimal fused estimate. If at instant $k$, $\hbox{max}(\bm{\lambda}_i^*) = \lambda_{ii}^*$, then the greatest contribution to $\bar{\mathbf{x}}_i^*$ is $\bar{\mathbf{x}}_i$, so the information at node $i$ is the best in its neighbourhood and should be broadcast. The result is the ET rule in Eq.~\eqref{eq:ETrule1}.

Some aspects must be underlined. On the one hand, at instant $k$, node $i$ receives $|\mathcal{N}_i(k)|$ messages from broadcasting neighbours. Thus, the network topology is time-varying.
On the other hand, two situations are possible in the ET rule at node $i$ and instant $k$: $|\mathcal{N}_i(k)|>0$ and $|\mathcal{N}_i(k)|=0$. The former is the core of the ET rule because it determines the communication usage. The latter, despite influencing in the bandwidth as well, must be designed to avoid that the network converges to $N$ individual KFs. As we will see later, this is key to assure stability.

\subsection{Theoretical analysis}\label{subsec:analysisET}

The two main properties analysed in Section~\ref{sec:stability} are stability and optimality. Regarding stability, the next Theorem shows that ECO-DKF is globally asymptotically stable under PJC.

\begin{theorem}\label{proposition:PJCstability}
ECO-DKF is global asymptotically stable for any time-varying graph sequence that satisfies PJC with $t < \infty$. 
\end{theorem}
\begin{proof}
The error dynamics $\eta^+$ in Theorem~\ref{theorem:stability} now depends on time
\begin{equation}
    \eta^+ = \mathbf{F}(k)\left.\mathbf{L}_w(k)\right. \mathbf{\mathcal{A}}\eta,
\end{equation}
From Theorem~\ref{theorem:stability}, since $\mathbf{L}_w(k)$ is row stochastic for any graph, the absolute value of its eigenvalues is always less than or equal to $1$.
Besides, $\mathbf{F}(k)$ is always a positive definite matrix and $\mathbf{I} \succeq \mathbf{F}(k)$ according to its definition in Lemma~\ref{lemma:matrices}.
Therefore, as a worst case scenario, we can consider $\mathbf{L}_w(k)=\mathbf{I}  \Rightarrow \eta^+ = \mathbf{F}(k) \mathbf{\mathcal{A}}\eta.$

Let $j$ be one of the nodes. The error dynamics, $\eta^+_j,$ are then determined by the eigenvalues and eigenvectors of $\mathbf{F}_j(k)\mathbf{A}$. 
The convergence rate of the KF is upper bounded by its steady-state gain $\mathbf{K}^{ss}$~\cite{Grewal2014Kalman}, so let instead consider
\begin{equation}\label{eq:some}
    \eta^+_j = \mathbf{F}_j^{ss}\mathbf{A}\eta_j.
\end{equation}
Let $\bm{v}_j = \{v_j^1, \hdots, v_j^r, \hdots, v_j^n\}$ and $\mathcal{V}_j = \{\mathbf{v}_j^1, \hdots, \mathbf{v}_j^r, \hdots, \mathbf{v}_j^n\}$ be the eigenvalues and eigenvectors of $\mathbf{F}_{j}^{ss}\mathbf{A}$.
The initial conditions are expressed as $\eta_j(0) = \sum_{r=1}^n\gamma_j^r\mathbf{v}_j^r$ and the norm of the error dynamics in Eq. \eqref{eq:some} can be upper bounded by \mbox{$||\eta^+_j(k)|| \leq \sum_{r=1}^n ||\left(v_j^r\right)^k \gamma_j^r\mathbf{v}_j^r||.$}
Network observability implies that for each node $j$ and eigenvalue $v_j^r$, there exist at least a node $j^{\prime}$ such that $\kappa|\Re(v_j^r)|<1$, with $\kappa \in \mathbb{R}_{>0}$ and $\Re(v_j^r)$ the real part of $v_j^r$.
Otherwise, it would not be possible to reconstruct the state, even if the network was fully connected. 
In the worst case, due to the PJC with period equal to $t$, the estimation of node $j^{\prime}$ associated to $v_j^r$ is updated through $t(N-1)$ ECO-DKF instants before reaching node $j$.
At each instant, ECO-DKF corrects the estimates in two main phases: the optimisation in~\eqref{eq:outerLJellipse} and the KF steps. The former consists in finding the intersection of $|\mathcal{N}_i(k) \cup \{i\}|$ ellipsoids. Since the intersection can not be larger than the smallest ellipsoid, then the result is upper bounded by the best estimate among estimates at instant $k$. Thus, in the worst case the estimation error associated to $v_j^r$ remains that of $j^{\prime}$. 
On the other hand, in the worst case, the estimation error associated to $v_j^r$ is only predicted in the KF steps without using the measurements to perform the correction step of the filter, and the estimation error associated to $v_j^r$ increases as a function of the eigenvalues of $\mathbf{A}$. Therefore, when the estimation of node $j^{\prime}$ associated to $v_j^r$ reaches node $j$, the error associated to $v_j^r$ is, in the worst case, the error for $j^{\prime}$ increased by a \emph{constant} value that depends on the eigenvalues of $\mathbf{A}$ and $t(N-1)$.

Since the estimation error associated to $v_j^r$ in $j^{\prime}$ tends to zero when $k \rightarrow \infty$, then so does the error in node $j$.
Applying this result to all the eigenvectors and nodes in the network, the result is proved.
\end{proof}

The next Proposition addresses optimality.
\begin{proposition}\label{proposition:PJCoptimality}
Theorem~\ref{theorem:optimality} holds under PJC and network-observability conditions. 
\end{proposition}
\begin{proof}
At instant $k$, node $i$ receives $|\mathcal{N}_i(k)|$ messages and conducts optimisation~\eqref{eq:outerLJellipse} and certification~\eqref{eq:traceRelax} steps. The case where $|\mathcal{N}_i(k)|=0$ is trivial and $\bar{\mathbf{x}}_i^*(k)=\bar{\mathbf{x}}_i(k)$ and $\bar{\mathbf{P}}_i^*(k)=\bar{\mathbf{P}}_i(k)$, according to restrictions 1 and 2 in Problem~\ref{problem:ECO-DKF}. In the general case, certifiability and optimality in Proposition~\ref{proposition:certifiability} and Theorem~\ref{theorem:optimality} are proved for any $|\mathcal{N}_i(k)|$. Therefore, Theorem~\ref{theorem:optimality} holds irrespective of the connectivity and observability of the network.
\end{proof}
The previous results are interesting not only because they demonstrate that ECO-DKF is still optimal and stable under PJC, but also because this can be done without requiring any additional information. 
Regarding the ET rule~\eqref{eq:ETrule1}, by letting each node to communicate with probability $p>0$ when it does not receive information, 
we force that nodes do not isolate forever, and ET rule~\eqref{eq:ETrule1} guarantees PJC with an expected value of $t$ bounded.

To complete the solution, we show that ECO-DKF is a consistent estimator under PJC conditions.
\begin{proposition}\label{proposition:consistency}
ECO-DKF is a consistent estimator under PJC conditions.
\end{proposition}
\begin{proof}
Recalling the definition of~\eqref{eq:outerLJellipse}, the outer LJ method gives a fused $\bar{\mathbf{S}}_i^*$ which is always consistent. Moreover, considering that we assume unbiased sensors, $\bar{\mathbf{x}}_i^*$ preserves consistency. This is because: (i) constraint~\eqref{eq:constraint32} ensures positive weights and their sum less or equal to $1$, so it is a non-increasing combination of estimates; (ii) Lemma~\ref{lemma:convergence} shows that the sum of the weights is always equal to $1$, so it is a non-decreasing combination of estimates. Thus, the aggregation of estimates in~\eqref{eq:KF_o1} is also unbiased. Besides, Theorem~\ref{proposition:PJCstability} states that ECO-DKF is globally asymptotically stable under PJC conditions, which means that once estimates are aggregated, the update and prediction steps preserve consistency as well. Therefore, the whole algorithm is consistent.
\end{proof}


\section{Illustrative results}\label{sec:examples}
We first evaluate the performance of the Time-Triggered (TT) ECO-DKF, i.e., ECO-DKF with $g_i(\bm{\lambda}^*_i(k-1),k)=1$ for all $i,k$, to verify the stability and optimality properties. Then, different ET rules are compared, studying their impact in the communication bandwidth and convergence speed.

\subsection{Time-triggered ECO-DKF vs Time-triggered  DKFs}\label{subsec:illustrative_example_TT}

We compare ECO-DKF in the TT setting with respect to other TT DKFs: (i) \textit{AtA-ECO-DKF}, our proposal but assuming that $(i,j)\in E(k)$ $\forall i\neq j, k$, (ii) \textit{OCDFK}, Algorithm 3 from~\cite{Olfati2007DKF} with consensus gain $\gamma_i(k) = 10^{-4}/(1 + ||\mathbf{M}_i(k)||_F)$ to ensure stability, with $||\cdot||_F$ the Frobenius norm, (iii) \textit{TCDFK}, Algorithm 2 from~\cite{Talebi2019Distributed} with one consensus step for a fair comparison, (iv) \textit{HDfKF}, complete algorithm in~\cite{Hu2011Diffusion}, (v) \textit{HADfKF}, simplified algorithm in~\cite{Hu2011Diffusion}, (vi) \textit{CKF}, centralised equivalent KF. The target system is the same 2D particle in a circular orbit used in~\cite{Sebastian2021CDC}, with the same parameters.

We analyse two scenarios. Experiment $1$ initialises a random sensor network, with appropriate parameters to obtain a sparse connected topology. Then, a random uniform distribution decides the quantities sensors measure among two options: measuring $x$ or $y$. Note that local observability never holds, to test ECO-DKF in a network-observability setting. $\mathbf{H}_i$ is picked uniformly in the range $[1,3]$. A Bernoulli distribution with $p=0.5$ decides the diagonal of $\mathbf{R}_i$ in the range $[3,5]\times 10^{-2}$ or $[3,5]$, i.e., high-quality or low-quality. This is done $100$ times, computing the averaged Mean Square Error MSE$:= \frac{1}{100} \sum_{i=1}^{N} \hbox{E}[||\hat{\mathbf{x}}_i-\mathbf{x}||^2]$ over the experiments as in~\cite{Cattivelli2010Diffusion} and~\cite{Hu2011Diffusion}. Experiment $2$ is the same as Experiment $1$, but only one sensor is of high quality, and sensors now can measure $x$, $y$ or both.

The results of Experiments $1$ are shown in Fig.~\ref{fig:RMSEvsREAL_transient}a. The best performance among the distributed estimators is obtained by ECO-DKF, with a difference of more than an order of magnitude with the other state-of-the-art filters in steady-state MSE. Besides, ECO-DKF is also the fastest filter, with the fastest asymptotic convergence to the CKF. The OCDKF achieves similar MSE performance than the AtA-ECO-DKF: adding more neighbours hinders the optimization problem in ECO-DKF. The other consensus-based DKF, TCDKF, is far from them because it needs multiple consensus steps within instants for a good performance~\cite{Talebi2019Distributed}. The diffusion-based DKF, HDfKF and HADfKF, exhibit worse suboptimal performance as well. These differences also hold in Experiment $2$, as it is shown in Fig.~\ref{fig:RMSEvsREAL_transient}b. Therefore, ECO-DKF obtains the best performance among DKFs. 

\begin{figure}[!ht]
\centering
\begin{tabular}{ccc}
\subcaptionbox{Experiment 1}{\includegraphics[width=0.35\columnwidth]{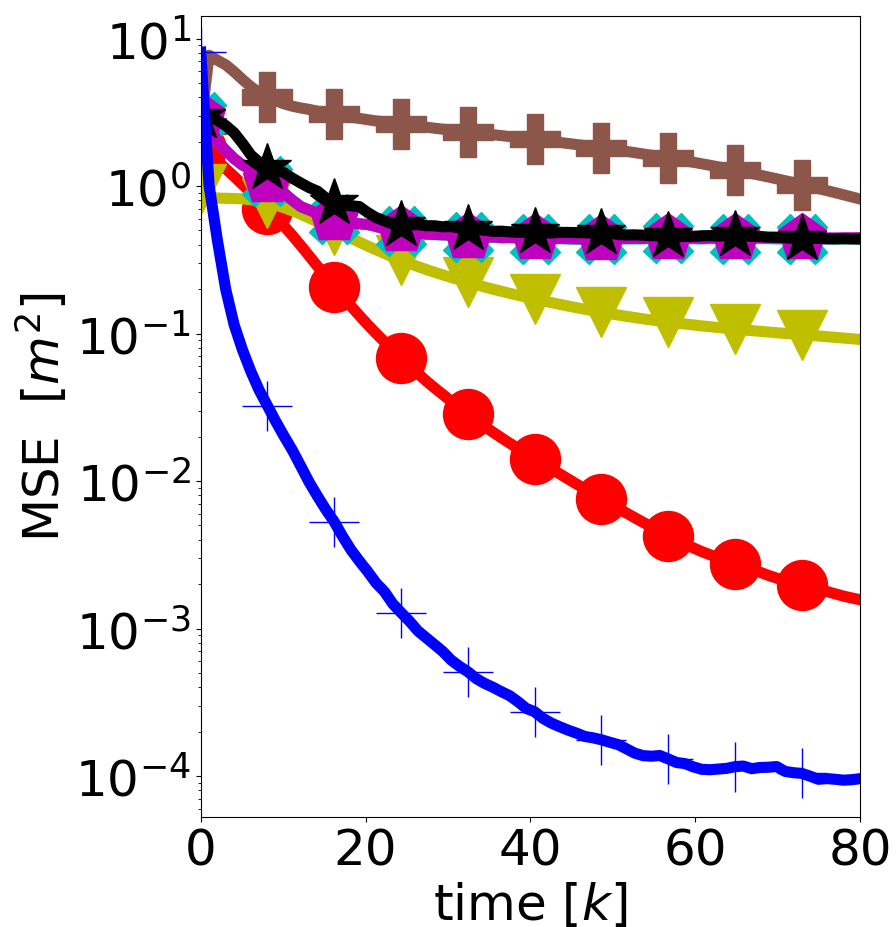}}
&
\subcaptionbox{Experiment 2}{\includegraphics[width=0.35\columnwidth]{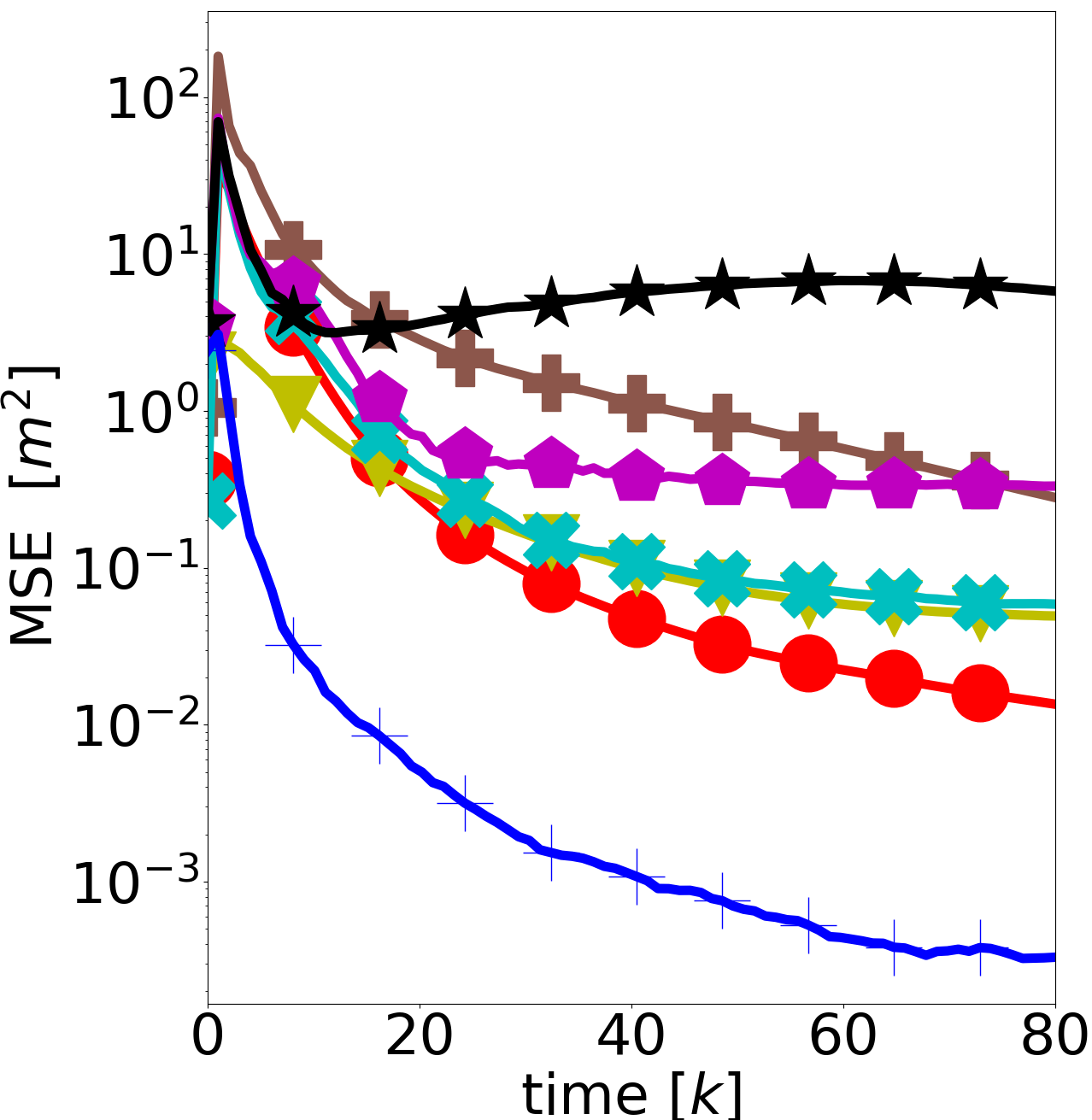}}
&
\includegraphics[width=0.17\columnwidth]{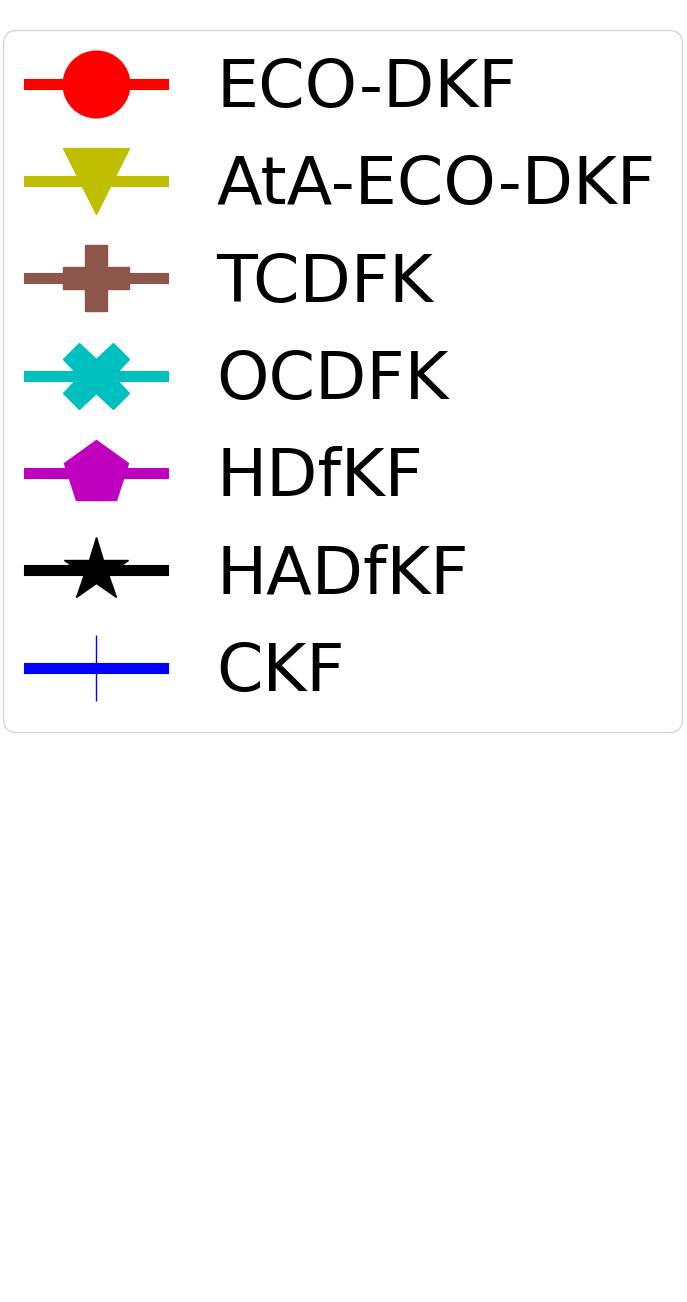}
\end{tabular}
\caption{{\small MSE for the different estimators in the TT case.}}
\label{fig:RMSEvsREAL_transient}
\end{figure}

We have assessed the tightness of the relaxation by checking the certification values. A $99.33\%$ of the times in Experiment $1$ and $95.92\%$ of the times in Experiment $2$, $\mathbb{C}_i = \rho_i = 1$, so the SDP relaxation is \textit{tight}. It is also noteworthy that the lower values of $\rho$ are at $0.65$, and this is only in the initialisation.

In relation with times, the computation of optimisations~\eqref{eq:outerLJellipse} and~\eqref{eq:traceRelax}, run at each iteration, takes $50$ms using a standard laptop. This is consistent with the election of the sample time $T=100$ms. Problem~\eqref{eq:outerLJellipse} is solved in a polynomial number of iterations and arithmetic operations~\cite{nesterov1994interior}, so the computational cost scales with the number of neighbours. Whether it is preferable to compute more or communicate more will depend on the network resources, but in general, communicating multiple times may imply a total run time and risk of disturbances greater than ours. \blue{In this sense, we recall that ECO-DKF only communicates, at most, once per instant $k$.}

\subsection{Evaluation of event-triggered schemes in  ECO-DKF}\label{subsec:illustrative_example_ET}

We now evaluate the impact of the ET rules in the NoB and MSE. We repeat Experiment $1$ and $2$ but with different ET rules: TT ECO-DKF (\#O), the original TT ECO-DKF tested in Subsection~\ref{subsec:illustrative_example_ET}; connected ECO-DKF (\#C), which corresponds to ET rule~\eqref{eq:ETrule1} with $p(k)=1$ $\forall k$; disconnected ECO-DKF (\#D), which corresponds to ET rule~\eqref{eq:ETrule1} with $p(k)=0$ $\forall k$; stochastic ECO-DKF (\#S), which corresponds to ET rule~\eqref{eq:ETrule1} with $p(k) \sim \mathcal{NBB}(\alpha,\beta,n,k_l,k)$. Here, $\mathcal{NBB}(\alpha,\beta,n,k_l,k)$ is a normalised beta-binomial distribution whose shape is given by $\alpha,\beta,n$. Given the shape, $p(k)$ increases when $k_l$ increases, where $k_l \in \mathbb{N}$ is the number of instants that have passed since the last time $g_i(\bm{\lambda}^*_i)=1$. The normalisation of the beta binomial is to force that, at some point, $p(k)=1$, i.e., to force PJC. Finally, we also test a Jensen-Shannon (JS) divergence ECO-DKF (\#J), which corresponds to a SoD rule based on the JS divergence
\begin{equation}\label{eq:ETrule3}
    g_i(\bm{\lambda}_i^*) = \left\{
\begin{aligned}
    1, & \hbox{ if } D_{JS}(\bar{\mathbf{S}}_i^*(k)||\bar{\mathbf{S}}_i^*(k-k')) > \tau 
    \\ 
    0, & \hbox{ otherwise }
\end{aligned}
\right. \forall i \in V
\end{equation}
where $\tau$ is a tuned threshold, $k'$ is the last time $g_i(\bm{\lambda}_i^*)=1$ and JS is a symmetric version of the KL divergence.
Note that ET rule~\eqref{eq:ETrule3} implicitly depends on $\bm{\lambda}_i^*$ since it is the divergence between the latest computation of~\eqref{eq:outerLJellipse} and the latest time node ET rule~\eqref{eq:ETrule3} was equal to $1$.  
For \#S, we have calibrated two combinations of parameters: $\alpha=\{5,2\}$, $\beta = \{1,2\}$ and $n=\{10,15\}$. For \#J we have empirically tuned the threshold to be $\tau = 100$. 

Figure~\ref{fig:NoB_ET} shows the evolution of NoB with time for the different ET rules, taking the average over the simulation runs. We can see that for \#O, $\hbox{NoB} = N$ always, since it is a TT version of ECO-DKF. In contrast, \#C reduces to a half the NoB, which is a significant decrease with a very simple tuning-free ET rule. As expected, \#D evolves to $0$ as the network tends to be disconnected. Regarding the stochastic rules (\#S), depending on the shape of the beta-binomial distribution the network converges to a different steady-state NoB. 
This also affects the NoB transient, in the sense that a sharper distribution at lower $k'$ values imply less time to trigger, reason why in \#S2 there is a very low NoB before $k=30$. 
Finally, the Jensen-Shannon trigger protocol achieves a trade-off between TT ECO-DKF (in the first instants) and a connected ECO-DKF (for $k>60$). 

\begin{figure}[!ht]
\centering
\begin{tabular}{ccc}
    \includegraphics[width=0.38\columnwidth]{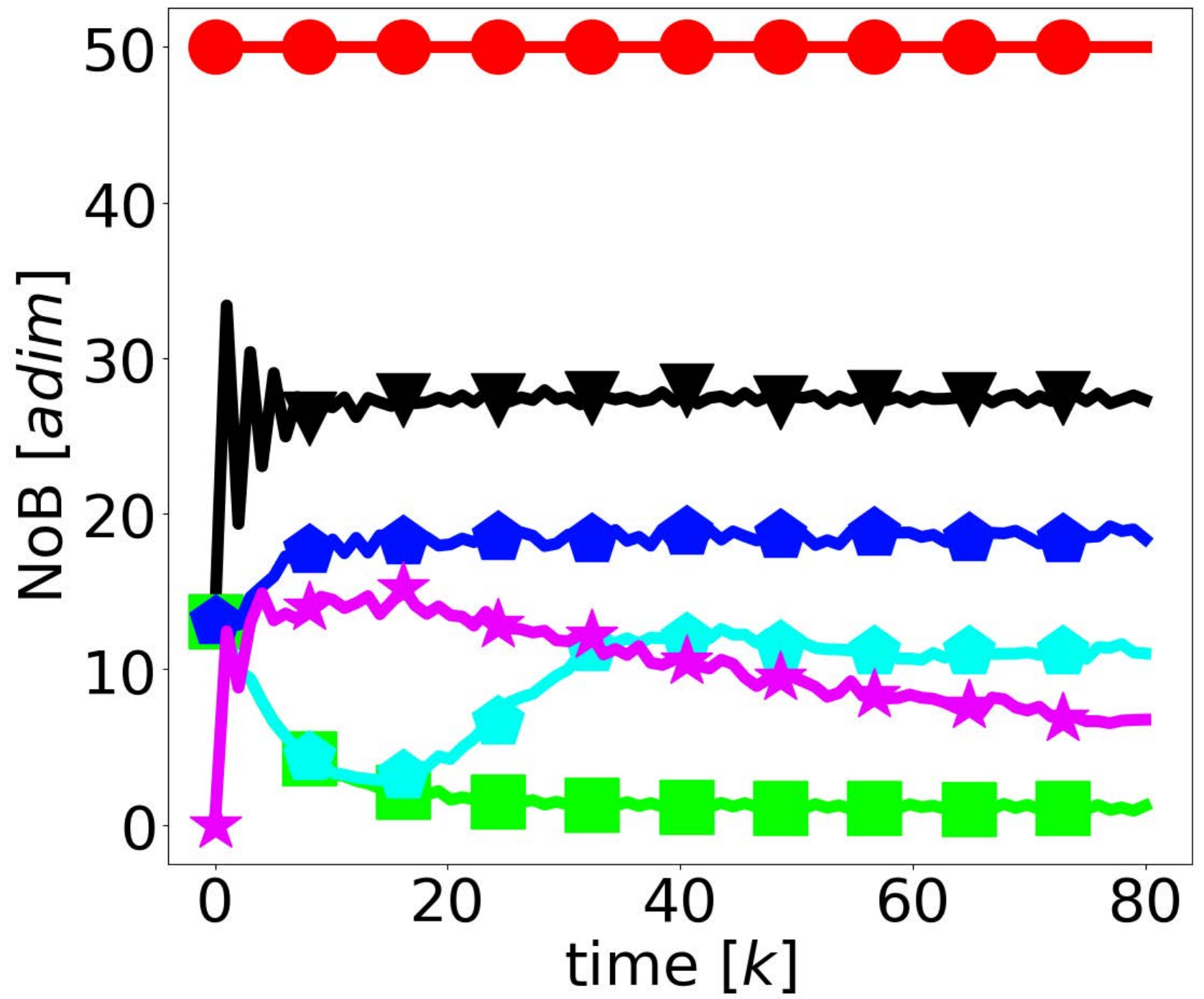}
    &
    \includegraphics[width=0.38\columnwidth]{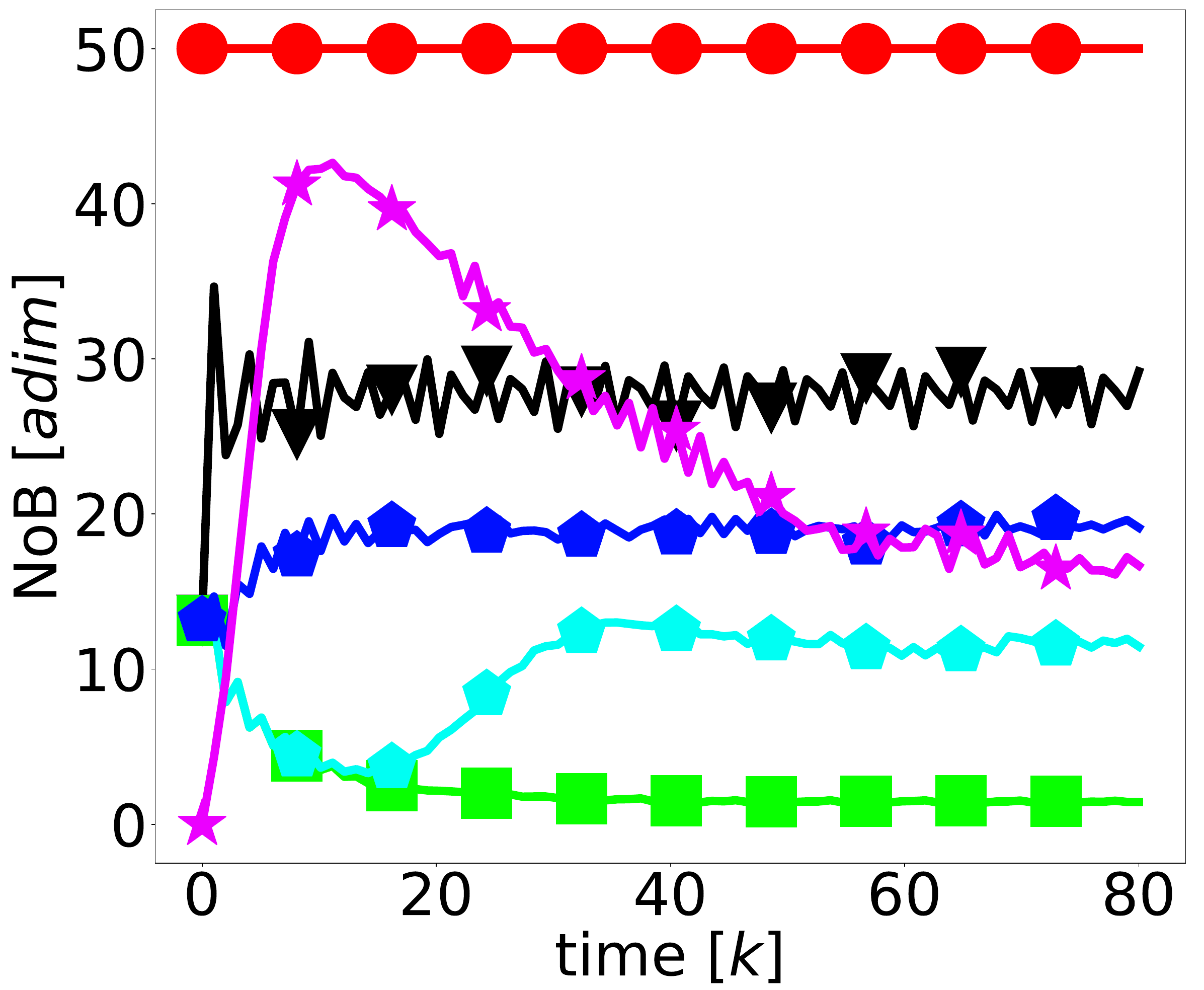}
    &
    \includegraphics[width=0.1\columnwidth]{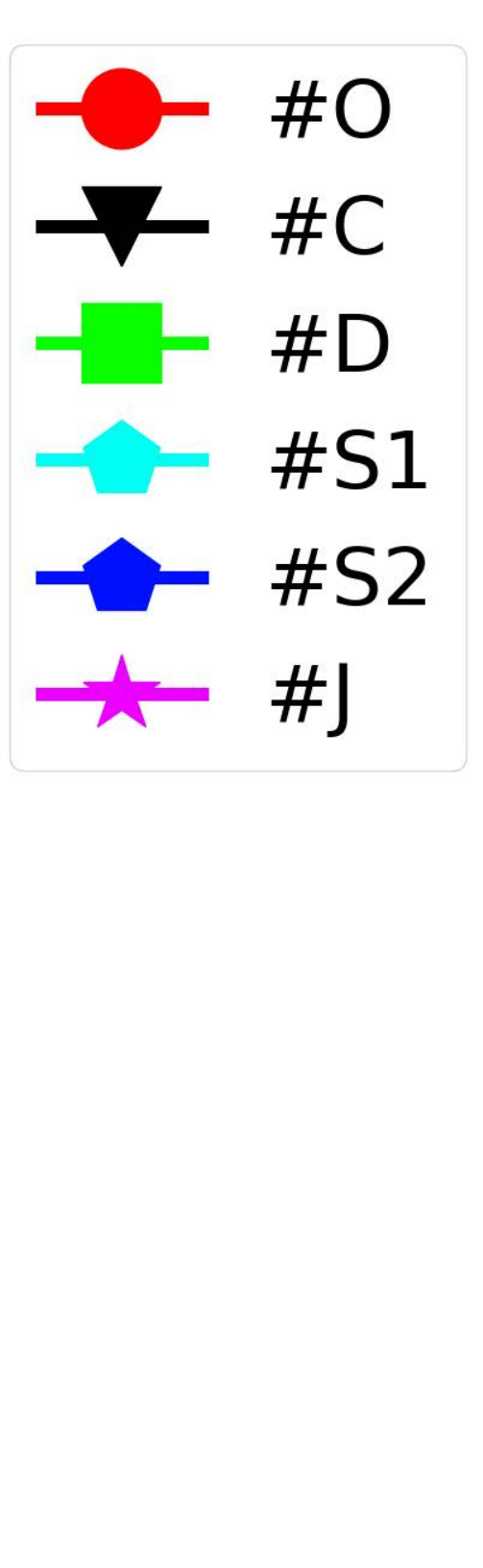}
\end{tabular}
\caption{{\small Evolution of NoB average over simulation runs with time for the different ET rules: (left) Experiment 1, (right) Experiment 2.}}
\label{fig:NoB_ET}
\end{figure}

Efficiency in communications is followed by a non-degraded performance. In Fig.~\ref{fig:RMSEvsREAL_transient_ET} we depict the MSE of the ET rules in Experiments $1$ and $2$. In the first instants, those protocols with more communications have the worst MSE (e.g., \#O) due to the poor initial estimates. However, afterwards, the tendency is inverted, (e.g., \#J). In steady-state, the improvement in communications given by \#S1 and \#S2 is followed by a slight improvement in MSE. \#C has a good trade-off between MSE and NoB. \#D is the worst estimator;  indeed, \#D failed $10$ runs in Experiment $1$ and $13$ runs in Experiment $2$. From Fig.~\ref{fig:RMSEvsREAL_transient}, the MSE of ECO-DKF is the closest to the CKF. In fact, given the high number of positive certifications, ECO-DKF is almost always optimal in the TT version. For general ET cases, ECO-DKF obtains almost the same performance compared to the TT version (Fig.~\ref{fig:RMSEvsREAL_transient_ET}), so the ET rule does not significantly affect the performance.

\begin{figure}[!ht]
\centering
\begin{tabular}{ccc}
\subcaptionbox{Experiment 1}{\includegraphics[width=0.35\columnwidth]{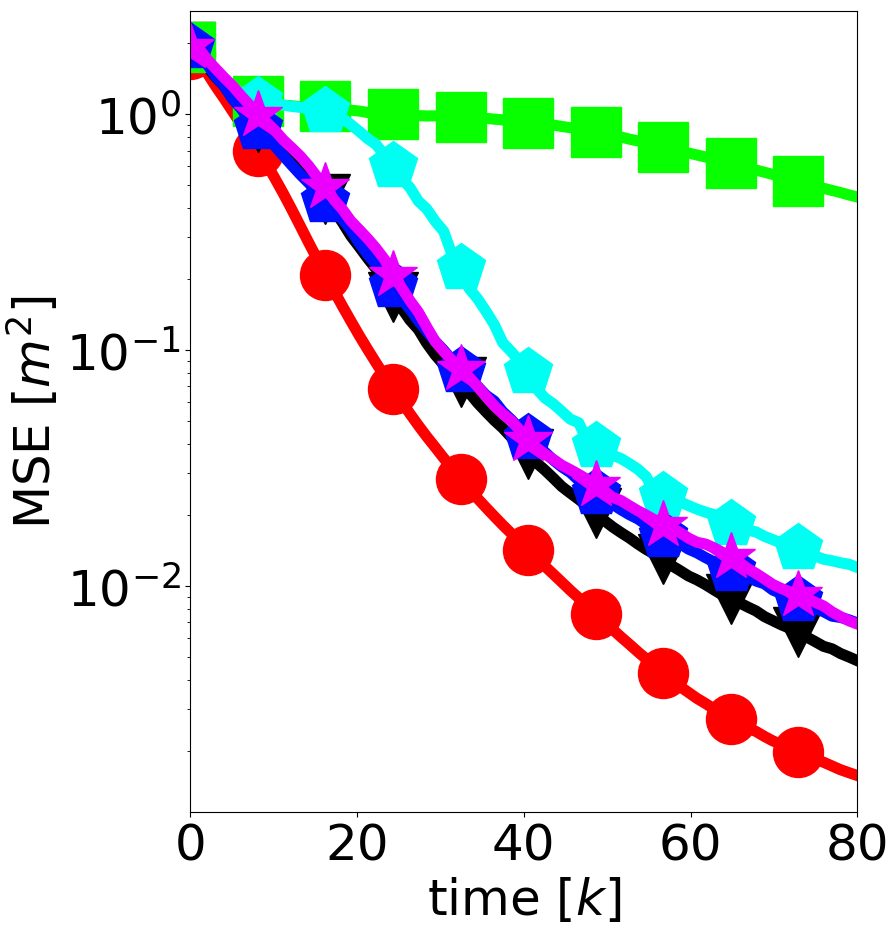}}
&
\subcaptionbox{Experiment 2}{\includegraphics[width=0.35\columnwidth]{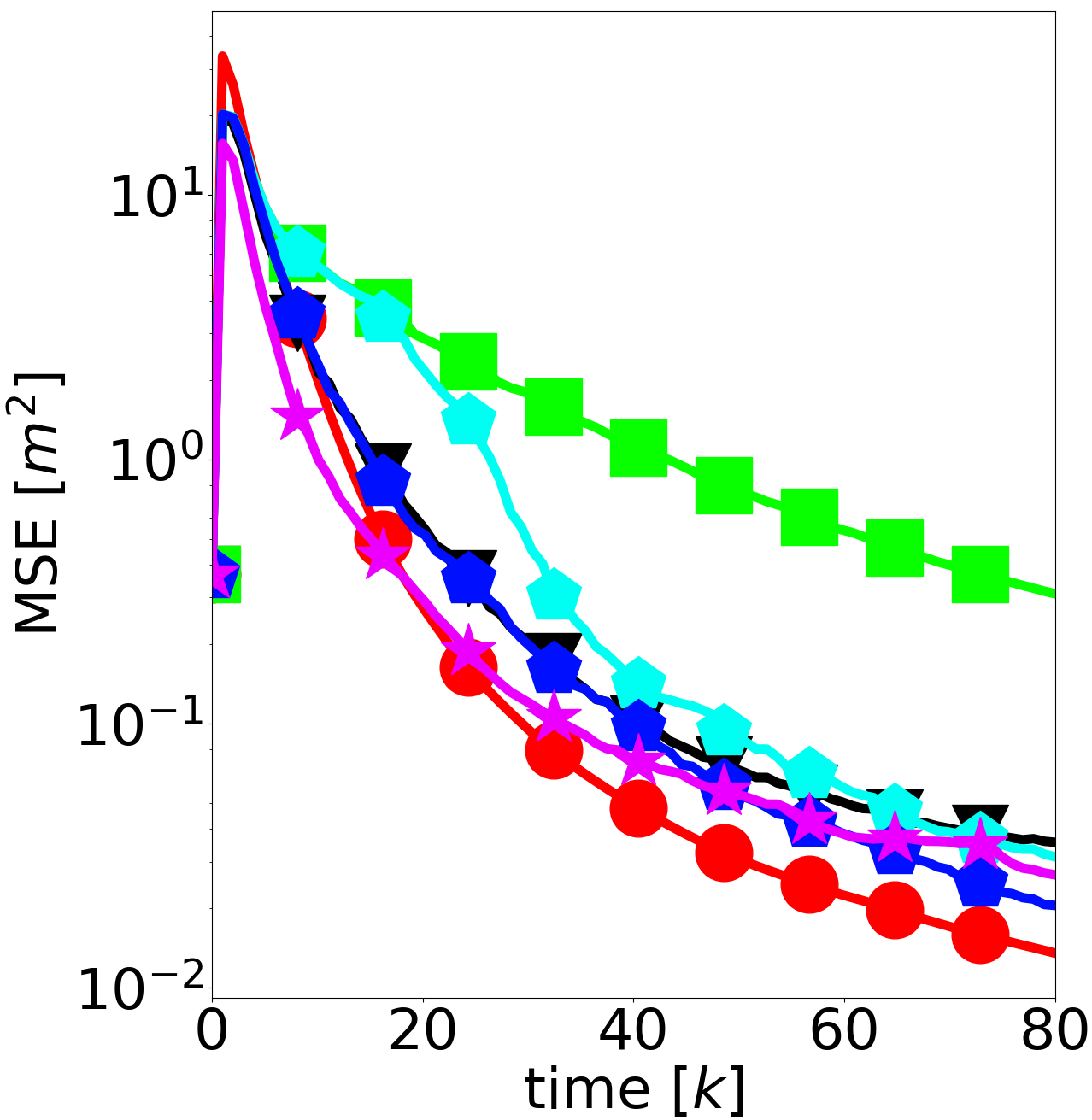}}
&
\includegraphics[width=0.1\columnwidth]{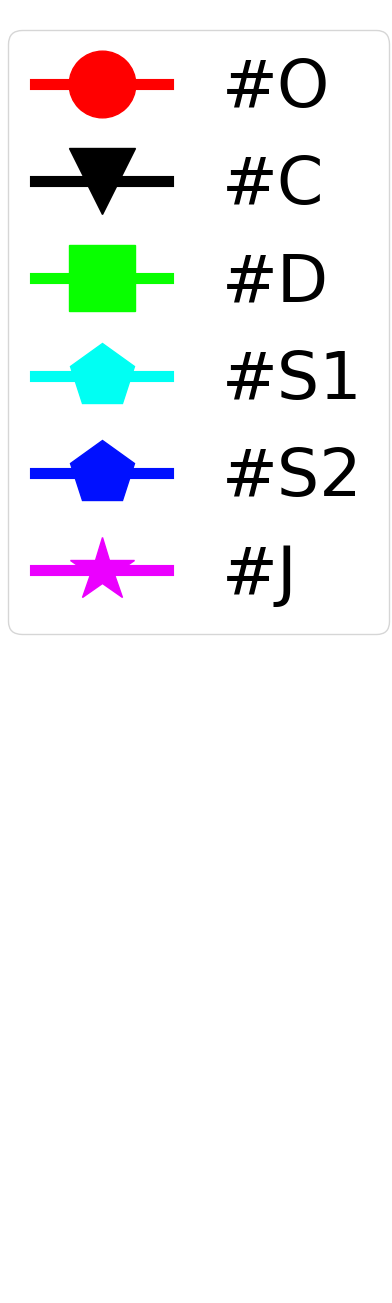}
\end{tabular}
\caption{{\small MSE for the different estimators in the ET case.}}
\label{fig:RMSEvsREAL_transient_ET}
\end{figure}

\vspace{-0.4cm}

\begin{table}[!ht]
\centering
\begin{tabular}{cc}
\:\:\:\:\:\:\:\:{\small Experiment $1$} 
& 
\:\:\:\:\:\:\:\:{\small Experiment $2$}
\\
\:\:\:\:\:\:\:\:\begin{tabular}{|c|c|}
\hline
ET rule & Avg($\mathbb{C}_i^f=\rho_i^f=1$)
\\
\hline
\#O & $99.33\%$ 
\\
\#C & $79.30\%$  
\\
\#D & $11.80\%$  
\\
\#S1 & $44.11\%$ 
\\
\#S2 & $69.88\%$ 
\\
\#J & $47.69\%$ 
\\
\hline
\end{tabular}
&
\:\:\:\:\:\:\:\:\begin{tabular}{|c|c|}
\hline
ET rule & Avg($\mathbb{C}_i^f=\rho_i^f=1$)
\\
\hline
\#O & $95.92\%$ 
\\
\#C & $77.74\%$  
\\
\#D & $12.11\%$  
\\
\#S1 & $44.14\%$ 
\\
\#S2 & $69.16\%$ 
\\
\#J & $67.24\%$ 
\\
\hline
\end{tabular}
\end{tabular}
\caption{{\small Certification results during the ET simulations.}}
\label{table:certifiability_ET}
\end{table}
Finally, in Table \ref{table:certifiability_ET} we evaluate certifiability and optimality. To do so, those cases where $|\mathcal{N}_i(k)|=0$ are removed from $\mathbb{C}_i$ and $\rho_i$ since they are trivial to compute, denoted by $\mathbb{C}_i^f$ and $\rho_i^f$ respectively. The results are averaged over the sensors of the network at each instant. Apart from \#O, the only ET rule which resists with reasonable certification levels is \#C: the disconnected rule hardly ever certifies as optimal; \#S1 and \#S2 have worse certification results; and \#J has a large variability among experiments, so it is hard to predict its certification performance.


\section{Conclusions}\label{sec:conclusion}
This paper has presented ECO-DKF, the first event-triggered and certifiable optimal DKF. It has solved two main problems regarding DKFs: optimality in the estimation and reduction of the communication bandwidth. The outer LJ fuses neighbouring estimates with certifiable guarantees of optimality. The output is integrated in an information DKF, achieving a stable filter on heterogeneous sensor networks, with minimal message size and no tuning. We have proved global asymptotic stability of the estimator and optimality under positive certification. Moreover, a novel ET theory has been derived from the outer LJ method output. Under PJC, ECO-DKF preserves the properties of the algorithm. The ET rule is inexpensive to compute, and avoids individual communication links and multiple information exchanges. ECO-DKF surpasses the state-of-the-art DKFs while decreasing the communication bandwidth usage.


\bibliographystyle{IEEEtran}
\bibliography{biblio}

\balance

\end{document}